\documentclass[runningheads]{llncs}

\makeatletter
\renewcommand{\paragraph}{\@startsection{paragraph}{6}{\z@}{2ex}{-0.7em}{\normalsize\bf}}
\makeatother

\usepackage{amsmath,amssymb,mathrsfs}
\usepackage{stmaryrd,multirow}
\usepackage{fancyhdr}
\usepackage{xspace}
\usepackage{mathtools}
\usepackage{cite}
\usepackage{enumitem}
\usepackage{hyperref}

\newcommand{\paul}[1]{{\color{black}#1}}

\usepackage[pdflatex,recompilepics=false]{gastex}
\ifpdf

\definecolor{mygreen}{RGB}{0.0,180,0.0}
\colorlet{colj}{gray}
\definecolor{coll}{RGB}{31,0,149}
\definecolor{colk}{RGB}{0,152,83}

\usepackage{amsmath,amssymb}
\usepackage{stmaryrd}
\usepackage{xspace}
\usepackage{mathtools}

\usepackage[T1]{fontenc}
\usepackage{graphicx}

\begin{document}
\title{High-Level Message Sequence Charts: Satisfiability and Realizability Revisited}
\titlerunning{HMSCs: Satisfiability and Realizability Revisited}

\author{Benedikt Bollig\inst{1}\orcidID{0000-0003-0985-6115} \and
Marie Fortin\inst{2}\orcidID{0000-0001-5278-0430} \and
Paul Gastin\inst{1}\orcidID{0000-0002-1313-7722}}
\authorrunning{B. Bollig, M. Fortin, and P. Gastin}

\institute{Université Paris-Saclay, CNRS, ENS Paris-Saclay, LMF, Gif-sur-Yvette, France\\
\email{\{bollig, gastin\}@lmf.cnrs.fr}\\
\and
Universit{\'e} Paris Cit{\'e}, CNRS, IRIF, Paris, France\\
\email{marie.fortin@irif.fr}\\
}
\maketitle
\begin{abstract}
Message sequence charts (MSCs) visually represent interactions in distributed systems that communicate through FIFO channels. High-level MSCs (HMSCs) extend MSCs with choice, concatenation, and iteration, allowing for the specification of complex behaviors. This paper revisits two classical problems for HMSCs: satisfiability and realizability. Satisfiability (also known as reachability or nonemptiness) asks whether there exists a path in the HMSC that gives rise to a valid behavior. Realizability concerns translating HMSCs into communicating finite-state machines to ensure correct system implementations.

While most positive results assume bounded channels, we introduce a class of HMSCs that allows for unbounded channels
while maintaining effective implementations. On the other hand, we show that the corresponding satisfiability problem is
still undecidable.

\end{abstract}


\newcommand{\Procs}{\mathscr{P}}
\newcommand{\Msg}{\mathit{Msg}}
\newcommand{\Act}{\mathit{Act}}
\newcommand{\Send}{\mathit{Send}}
\newcommand{\Rec}{\mathit{Rec}}
\newcommand{\Loc}{\mathit{Loc}}
\newcommand{\Ch}{\mathit{Ch}}
\newcommand{\dom}{\mathsf{dom}}
\newcommand{\mEvents}{{E}}
\newcommand{\tEvents}{E}
\newcommand{\mlabel}{\lambda}
\newcommand{\lab}{\lambda}
\newcommand{\loclabel}{\ell}
\newcommand{\tlabel}{\lambda}
\newcommand{\proc}{\mathsf{proc}}
\newcommand{\procsim}{\sim_{\proc}}
\newcommand{\msgsim}{\sim_{\mathsf{msg}}}
\newcommand{\procrel}{\to}
\newcommand{\mprocrel}{\to}
\newcommand{\mprocle}{\le_{\proc}}
\newcommand{\mproclt}{<_{\proc}}
\newcommand{\tprocrel}{\to}
\newcommand{\tprocle}{\le}
\newcommand{\msgrel}{\lhd}
\renewcommand{\phi}{\varphi}
\newcommand{\compl}[1]{\llbracket #1 \rrbracket}
\newcommand{\tEMSO}{\textup{EMSO}_{\mathsf{trace}}}
\newcommand{\mEMSO}{\textup{EMSO}_{\mathsf{msc}}}
\newcommand{\tFO}{\textup{FO}_{\mathsf{trace}}}
\newcommand{\mFO}{\textup{FO}_{\mathsf{msc}}}
\newcommand{\mscof}[1]{\mathit{msc}(#1)}
\newcommand{\unm}[1]{\mathsf{unm}(#1)}
\newcommand{\loc}[1]{#1_\mathsf{local}}
\newcommand{\matched}[1]{#1_\mathsf{matched}}
\newcommand{\pending}[1]{#1_\mathsf{pending}}
\newcommand{\trace}{T}
\newcommand{\msc}{M}
\newcommand{\msca}{M}
\newcommand{\mscb}{N}
\newcommand{\mscconc}{\circ}
\newcommand{\HMSC}{\mathcal{H}}
\newcommand{\CFM}{\mathcal{A}}

\newcommand{\labproj}[3]{#1_{|(#2,#3)}}
\newcommand{\cmscproj}[2]{#1_{|#2}}

\newcommand{\mscrun}[1]{\mathit{cmscs}(#1)}
\newcommand{\mscrunH}[2]{\mathit{cmscs}_{#1}(#2)}

\newcommand{\MSCs}{\mathbb{MSC}}
\newcommand{\fpMSCs}{\textup{fp}\mathbb{MSC}}

\newcommand{\cMSCs}{\textup{c}\mathbb{MSC}}
\newcommand{\fMSCs}{\textup{f}\mathbb{MSC}}
\newcommand{\pMSCs}{\textup{c}\mathbb{MSC}}
\newcommand{\fcMSCs}{\textup{fc}\mathbb{MSC}}

\newcommand{\send}[2]{#1!#2}
\newcommand{\rec}[2]{#1?#2}
\newcommand{\locact}[1]{\mathit{loc}(#1)}

\newcommand{\Labels}{\mathit{Labels}}
\newcommand{\Attr}{\mathit{Attr}}
\newcommand{\Props}{\Lambda}

\newcommand{\expr}{\mathcal{E}}
\newcommand{\exprb}{\mathcal{F}}

\newcommand\labeq[2]{#1 \equiv_{Z,W}  #2}
\newcommand\proceq[2]{#1 \equiv_Z  #2}
\newcommand\msgeq[2]{#1 \equiv_Y  #2}
\newcommand\Weq[2]{#1 \equiv_W #2}
\newcommand\bord[2]{[#1] \rightsquigarrow [#2]}
\newcommand\bordn{{\rightsquigarrow}}
\newcommand\phir[1]{#1^{{}\sim x}}
\newcommand\phiw[1]{#1^{W}}
\newcommand\phinw[1]{#1^{\neg W}}

\newcommand{\leftend}{{\triangleright}} 
\newcommand{\blank}{\flat} 

\begin{gpicture}[name=Ma,ignore]
  \gasset{Nframe=n,Nfill=n,AHnb=0}
  \node(p)(  0,3){$p$} \drawline(  0,0)(  0,-11)
  \node(q)(-10,3){$q$} \drawline(-10,0)(-10,-11)
  \node(r)( 10,3){$r$} \drawline( 10,0)( 10,-11)
  \node(M)(0,-15){$M_{a}^{!}$}
  \gasset{Nw=1.6,Nh=1.6,Nfill=y,ExtNL=y,NLdist=0.8,AHnb=1}
	\node[Nmarks=f,fangle= -30,NLangle= 20](0)(0,-7){$a$}
	\node[Nmarks=f,fangle=-150,NLangle=160](0)(0,-3){$a$}
\end{gpicture}
\begin{gpicture}[name=Mb,ignore]
  \gasset{Nframe=n,Nfill=n,AHnb=0}
  \node(p)(  0,3){$p$} \drawline(  0,0)(  0,-11)
  \node(q)(-10,3){$q$} \drawline(-10,0)(-10,-11)
  \node(r)( 10,3){$r$} \drawline( 10,0)( 10,-11)
  \node(M)(0,-15){$M_{b}^{?}$}
  \gasset{Nw=1.6,Nh=1.6,Nfill=y,ExtNL=y,NLdist=0.8,AHnb=1}
	\node[Nmarks=i,iangle= 30,NLangle= -20](0)(0,-8){$b$}
	\node[Nmarks=i,iangle=150,NLangle=-160](0)(0,-4){$b$}
\end{gpicture}
\begin{gpicture}[name=Maf,ignore]
  \gasset{Nframe=n,Nfill=n,AHnb=0}
  \node(p)(  0,3){$p$} \drawline(  0,0)(  0,-16)
  \node(q)(-10,3){$q$} \drawline(-10,0)(-10,-16)
  \node(r)( 10,3){$r$} \drawline( 10,0)( 10,-16)
  \node(M)(0,-20){$M_{a}^{f}$}
  \gasset{Nw=1.6,Nh=1.6,Nfill=y,ExtNL=y,NLdist=1.0,AHnb=1}
	\node[Nmarks=i,iangle= 30,NLangle=-10](0)(-10, -3){$a$}
	\node[Nmarks=f,fangle=-30,NLangle=  0](0)(-10, -7){}
	\node[Nmarks=f,fangle=-30,NLangle=180](0)(-10,-10){$f(a)\Big\{$}
	\node[Nmarks=f,fangle=-30,NLangle=  0](0)(-10,-13){}
\end{gpicture}
\begin{gpicture}[name=Mag,ignore]
  \gasset{Nframe=n,Nfill=n,AHnb=0}
  \node(p)(  0,3){$p$} \drawline(  0,0)(  0,-16)
  \node(q)(-10,3){$q$} \drawline(-10,0)(-10,-16)
  \node(r)( 10,3){$r$} \drawline( 10,0)( 10,-16)
  \node(M)(0,-20){$M_{a}^{g}$}
  \gasset{Nw=1.6,Nh=1.6,Nfill=y,ExtNL=y,NLdist=1.0,AHnb=1}
	\node[Nmarks=i,iangle=150,NLangle=-170](0)(10, -3){$a$}
	\node[Nmarks=f,fangle=-150,NLangle=  0](0)(10, -7){}
	\node[Nmarks=f,fangle=-150,NLangle=  0](0)(10,-10){$\Big\}g(a)$}
	\node[Nmarks=f,fangle=-150,NLangle=  0](0)(10,-13){}
\end{gpicture}
\begin{gpicture}[name=flat-HMSC-PCP,ignore]
  \gasset{Nw=6,Nh=6,loopdiam=6,loopwidth=5}
  \node[Nmarks=i](1)(0,0){1}
  \node(2)(20,0){2}
  \node(3)(40,0){3}
  \node(4)(60,0){3}
  \node[Nmarks=f](5)(80,0){5}
  \drawloop(1){$M_{a}^{!}$}
  \drawedge(1,2){$M_{a}^{!}$}
  \drawloop(2){$M_{a}^{f}$}
  \drawedge(2,3){$M_{a}^{f}$}
  \drawloop(3){$M_{a}^{g}$}
  \drawedge(3,4){$M_{a}^{g}$}
  \drawloop(4){$M_{b}^{?}$}
  \drawedge(4,5){$M_{b}^{?}$}
\end{gpicture}
\begin{gpicture}[name=MrunTM,ignore]
  \gasset{Nframe=n,Nfill=n,AHnb=0}
  \unitlength=2mm
  \node(p)(-2,8){$p$} \drawline(0,8)(61,8)
  \node(q)(-2,0){$q$} \drawline(0,0)(61,0)
  \node(C)( 6.5,11.5){$\textcolor{blue}{C_{0}}$}
  \node(C)(18.5,11.5){$\textcolor{red}{C_{1}}=\textcolor{blue}{C_{1}}$}
  \node(C)(30.5,11.5){$\textcolor{red}{C_{2}}=\textcolor{blue}{C_{2}}$}
  \node(C)(42.5,11.5){$\textcolor{red}{C_{3}}=\textcolor{blue}{C_{3}}$}
  \node(C)(54.5,11.5){$\textcolor{red}{C_{4}}$}
  \multiput(0,0)(12,0){5}{%
    \node(p)(6.7,8.8){$\cdots$}
    \node(C)(6.5,9.5){$\overbrace{\hspace{19mm}}$}
  }
  \node(C)(17.5,-3.5){$\textcolor{blue}{C_{0}}\vdash\textcolor{red}{C_{1}}$}
  \node(C)(29.5,-3.5){$\textcolor{blue}{C_{1}}\vdash\textcolor{red}{C_{2}}$}
  \node(C)(41.5,-3.5){$\textcolor{blue}{C_{2}}\vdash\textcolor{red}{C_{3}}$}
  \node(C)(53.5,-3.5){$\textcolor{blue}{C_{3}}\vdash\textcolor{red}{C_{4}}$}
  \multiput(0,0)(12,0){4}{%
    \node(q)(17.7,-0.8){$\cdots$}
    \node(C)(17.5,-1.5){$\underbrace{\hspace{19mm}}$}
  }
  \gasset{Nw=0.6,Nh=0.6,Nfill=y,ExtNL=y,NLdist=0.5,AHnb=1}
  \multiput(0,0)(12,0){5}{%
    \node(p)( 1,8){}\node(q)(11,0){}\drawedge[linewidth=0.14](p,q){}
    \node(p)(12,8){}\node(q)(12,0){}\drawedge[linewidth=0.14](q,p){}
  }
  \gasset{fillcolor=blue,linecolor=blue}
  \multiput(0,0)(12,0){4}{%
    \node(p)( 3,8){}\node(q)(13,0){}\drawedge(p,q){}
    \node(p)( 5,8){}\node(q)(15,0){}\drawedge(p,q){}
    \node(p)( 9,8){}\node(q)(19,0){}\drawedge(p,q){}
    \node(p)(11,8){}\node(q)(21,0){}\drawedge(p,q){}
  }
  \gasset{fillcolor=red,linecolor=red}
  \multiput(12,0)(12,0){4}{%
    \node(p)(2,8){}\node(q)(2,0){}\drawedge(q,p){}
    \node(p)(4,8){}\node(q)(4,0){}\drawedge(q,p){}
    \node(p)(8,8){}\node(q)(8,0){}\drawedge(q,p){}
    \node(p)(10,8){}\node(q)(10,0){}\drawedge(q,p){}
  }
\end{gpicture}
\begin{gpicture}[name=Msharpi,ignore]
  \gasset{Nframe=n,Nfill=n,AHnb=0}
  \unitlength=2mm
  \node(p)(-1,4){$p$} \drawline(0,4)(4,4)
  \node(q)(-1,0){$q$} \drawline(0,0)(4,0)
  \node(M)(-5,2){$M_{\#}^{i}$}
  \gasset{Nw=0.6,Nh=0.6,Nfill=y,ExtNL=y,NLdist=1.2,AHnb=1}
	\node[Nmarks=f,fangle=-45,NLangle=-70](0)(1,4){$\scriptstyle{\#}$}
\end{gpicture}
\begin{gpicture}[name=Msharp,ignore]
  \gasset{Nframe=n,Nfill=n,AHnb=0}
  \unitlength=2mm
  \node(p)(-1,4){$p$} \drawline(0,4)(8,4)
  \node(q)(-1,0){$q$} \drawline(0,0)(8,0)
  \node(M)(-5,2){$M_{\#}$}
  \gasset{Nw=0.6,Nh=0.6,Nfill=y,ExtNL=y,NLdist=1.2,AHnb=1}
	\node[Nmarks=i,iangle=135,NLangle=110](0)(3,0){$\scriptstyle{\#}$}
  \node(p)(4,4){}\node(q)(4,0){}\drawedge(q,p){}
	\node[Nmarks=f,fangle=-45,NLangle=-70](0)(5,4){$\scriptstyle{\#}$}
\end{gpicture}
\begin{gpicture}[name=Msharpf,ignore]
  \gasset{Nframe=n,Nfill=n,AHnb=0}
  \unitlength=2mm
  \node(p)(-1,4){$p$} \drawline(0,4)(5,4)
  \node(q)(-1,0){$q$} \drawline(0,0)(5,0)
  \node(M)(-5,2){$M_{\#}^{f}$}
  \gasset{Nw=0.6,Nh=0.6,Nfill=y,ExtNL=y,NLdist=1.2,AHnb=1}
	\node[Nmarks=i,iangle=135,NLangle=110](0)(3,0){$\scriptstyle{\#}$}
  \node(p)(4,4){}\node(q)(4,0){}\drawedge(q,p){}
\end{gpicture}
\begin{gpicture}[name=Mgammai,ignore]
  \gasset{Nframe=n,Nfill=n,AHnb=0}
  \unitlength=2mm
  \node(p)(-1,4){$p$} \drawline(0,4)(4,4)
  \node(q)(-1,0){$q$} \drawline(0,0)(4,0)
  \node(M)(-5,2){$M_{\gamma}^{i}$}
  \gasset{Nw=0.6,Nh=0.6,Nfill=y,ExtNL=y,NLdist=1.2,AHnb=1}
  \gasset{fillcolor=blue,linecolor=blue}
	\node[Nmarks=f,fangle=-45,NLangle=-70](0)(1,4){\color{blue}$\scriptstyle{\gamma}$}
\end{gpicture}
\begin{gpicture}[name=Mgammap,ignore]
  \gasset{Nframe=n,Nfill=n,AHnb=0}
  \unitlength=2mm
  \node(p)(-1,4){$p$} \drawline(0,4)(6,4)
  \node(q)(-1,0){$q$} \drawline(0,0)(6,0)
  \node(M)(-5,2){$M_{\gamma}^{p}$}
  \gasset{Nw=0.6,Nh=0.6,Nfill=y,ExtNL=y,NLdist=1.2,AHnb=1}
  \gasset{fillcolor=red,linecolor=red}
	\node[Nmarks=i,iangle=-90,NLangle=-110](0)(2,4){\color{red}$\scriptstyle{\gamma}$}
  \gasset{fillcolor=blue,linecolor=blue}
	\node[Nmarks=f,fangle=-45,NLangle= -70](0)(3,4){\color{blue}$\scriptstyle{\gamma}$}
\end{gpicture}
\begin{gpicture}[name=Mgammaf,ignore]
  \gasset{Nframe=n,Nfill=n,AHnb=0}
  \unitlength=2mm
  \node(p)(-1,4){$p$} \drawline(0,4)(4,4)
  \node(q)(-1,0){$q$} \drawline(0,0)(4,0)
  \node(M)(-5,2){$M_{\gamma}^{f}$}
  \gasset{Nw=0.6,Nh=0.6,Nfill=y,ExtNL=y,NLdist=1.2,AHnb=1}
  \gasset{fillcolor=red,linecolor=red}
	\node[Nmarks=i,iangle=-90,NLangle=-110](0)(2,4){\color{red}$\scriptstyle{\gamma}$}
\end{gpicture}
\begin{gpicture}[name=Mgammaq,ignore]
  \gasset{Nframe=n,Nfill=n,AHnb=0}
  \unitlength=2mm
  \node(p)(-1,4){$p$} \drawline(0,4)(6,4)
  \node(q)(-1,0){$q$} \drawline(0,0)(6,0)
  \node(M)(-5,2){$M_{\gamma}^{q}$}
  \gasset{Nw=0.6,Nh=0.6,Nfill=y,ExtNL=y,NLdist=1.2,AHnb=1}
  \gasset{fillcolor=blue,linecolor=blue}
	\node[Nmarks=i,iangle=135,NLangle=110](0)(3,0){\color{blue}$\scriptstyle{\gamma}$}
  \gasset{fillcolor=red,linecolor=red}
	\node[Nmarks=f,fangle=90,NLangle=60](0)(4,0){\color{red}$\scriptstyle{\gamma}$}
\end{gpicture}
\begin{gpicture}[name=Mdelta,ignore]
  \gasset{Nframe=n,Nfill=n,AHnb=0}
  \unitlength=2mm
  \node(p)(-1,4){$p$} \drawline(0,4)(18,4)
  \node(q)(-1,0){$q$} \drawline(0,0)(18,0)
  \node(M)(-5,2){$M_{\delta}^{q}$}
  \gasset{Nw=0.6,Nh=0.6,Nfill=y,ExtNL=y,NLdist=1.2,AHnb=1}
  \gasset{fillcolor=blue,linecolor=blue}
	\node[Nmarks=i,iangle=135,NLangle=110](0)(3,0){\color{blue}$\scriptstyle{\alpha_{1}}$}
	\node[Nmarks=i,iangle=135,NLangle=110](0)(9,0){\color{blue}$\scriptstyle{\alpha_{2}}$}
	\node[Nmarks=i,iangle=135,NLangle=110](0)(15,0){\color{blue}$\scriptstyle{\alpha_{3}}$}
  \gasset{fillcolor=red,linecolor=red}
	\node[Nmarks=f,fangle=90,NLangle=60](0)(4,0){\color{red}$\scriptstyle{\beta_{1}}$}
	\node[Nmarks=f,fangle=90,NLangle=60](0)(10,0){\color{red}$\scriptstyle{\beta_{2}}$}
	\node[Nmarks=f,fangle=90,NLangle=60](0)(16,0){\color{red}$\scriptstyle{\beta_{3}}$}
\end{gpicture}

\begin{gpicture}[name=cMSC-example-M1,ignore]
  \gasset{Nframe=n,Nfill=n,AHnb=0}
  \node(p)( 0,3){$p$} \drawline(  0,0)(  0,-15)
  \node(q)(10,3){$q$} \drawline(10,0)(10,-15)
  \node(M)(5,-20){$M_1$}
  \node(e1)(-3,-3){$e$}
  \node(e1)(13,-6){$f$}
  \node(e1)(-3,-10){$e'$}
  \gasset{Nw=1.6,Nh=1.6,Nfill=y,ExtNL=y,NLdist=0.8,AHnb=1}
	\node[Nmarks=f,fangle= -20,NLangle= 20](0)(0,-10){$a$}
    \node(p)(0,-3){}\node(q)(10,-6){}\drawedge(p,q){}
\end{gpicture}
\begin{gpicture}[name=cMSC-example-M2,ignore]
  \gasset{Nframe=n,Nfill=n,AHnb=0}
  \node(p)( 0,3){$p$} \drawline(  0,0)(  0,-15)
  \node(q)(10,3){$q$} \drawline(10,0)(10,-15)
  \node(M)(5,-20){$M_2$}%
  \gasset{Nw=1.6,Nh=1.6,Nfill=y,ExtNL=y,NLdist=0.8,AHnb=1}
	\node[Nmarks=f,fangle= -20,NLangle= 20](0)(0,-6){$a$}
\end{gpicture}
\begin{gpicture}[name=cMSC-example-M3,ignore]
  \gasset{Nframe=n,Nfill=n,AHnb=0}
  \node(p)( 0,3){$p$} \drawline(  0,0)(  0,-15)
  \node(q)(10,3){$q$} \drawline(10,0)(10,-15)
  \node(M)(5,-20){$M_3$}
  \gasset{Nw=1.6,Nh=1.6,Nfill=y,ExtNL=y,NLdist=0.8,AHnb=1}
 \node[Nmarks=i,iangle=160,NLangle=-170](0)(10, -10){$a~$}
\end{gpicture}
\begin{gpicture}[name=cMSC-example-M4,ignore]
  \gasset{Nframe=n,Nfill=n,AHnb=0}
  \node(p)( 0,3){$p$} \drawline(  0,0)(  0,-15)
  \node(q)(10,3){$q$} \drawline(10,0)(10,-15)
  \node(M)(5,-20){$M_4$}%
  \gasset{Nw=1.6,Nh=1.6,Nfill=y,ExtNL=y,NLdist=0.8,AHnb=1}
    \node(p)(0,-3){}\node(q)(10,-6){}\drawedge(p,q){}
\end{gpicture}
\begin{gpicture}[name=cMSC-example-M5,ignore]
  \gasset{Nframe=n,Nfill=n,AHnb=0}
  \node(p)( 0,3){$p$} \drawline(  0,0)(  0,-15)
  \node(q)(10,3){$q$} \drawline(10,0)(10,-15)
  \node(M)(5,-20){$M_5$}
  \gasset{Nw=1.6,Nh=1.6,Nfill=y,ExtNL=y,NLdist=0.8,AHnb=1}
\node[Nmarks=f,fangle= -20,NLangle= 20](0)(0,-10){$a$}
 \node[Nmarks=i,iangle=160,NLangle=-170](0)(10, -6){$a~$}
\end{gpicture}%
\begin{gpicture}[name=cMSC-example-M6,ignore]
  \gasset{Nframe=n,Nfill=n,AHnb=0}
  \node(p)( 0,3){$p$} \drawline(  0,0)(  0,-15)
  \node(q)(10,3){$q$} \drawline(10,0)(10,-15)
  \node(M)(5,-20){$M_6$}
  \gasset{Nw=1.6,Nh=1.6,Nfill=y,ExtNL=y,NLdist=0.8,AHnb=1}
    \node(p)(0,-3){}\node(q)(10,-6){}\drawedge(p,q){}
    \node(p)(0,-10){}\node(q)(10,-13){}\drawedge(p,q){}
\end{gpicture}

\begin{gpicture}[name=HMSC-example-M,ignore]
  \gasset{Nframe=n,Nfill=n,AHnb=0}
  \node(p)( 0,3){$p$} \drawline(  0,0)(  0,-12)
  \node(q)(10,3){$q$} \drawline(10,0)(10,-12)
  \node(M)(5,-16){$M$}
  \gasset{Nw=1.6,Nh=1.6,Nfill=y,ExtNL=y,NLdist=0.8,AHnb=1}
  \node(p)(0,0){}\node(q)(10,-6){}\drawedge(p,q){}
  \node(p)(0,-3){}\node(q)(10,-9){}\drawedge(p,q){}
  \node(p)(0,-6){}\node(q)(10,-12){}\drawedge(p,q){}
  
  \node(p)(10,-0){}\node(q)(0,-9){}\drawedge(p,q){}
  \node(p)(10,-3){}\node(q)(0,-12){}\drawedge(p,q){}
\end{gpicture}
\begin{gpicture}[name=HMSC-example-M1,ignore]
  \gasset{Nframe=n,Nfill=n,AHnb=0}
  \node(p)( 0,3){$p$} \drawline(  0,0)(  0,-12)
  \node(q)(10,3){$q$} \drawline(10,0)(10,-12)
  \node(M)(5,-15){$M_p$}
  \gasset{Nw=1.6,Nh=1.6,Nfill=y,ExtNL=y,NLdist=0.8,AHnb=1}
	\node[Nmarks=f,fangle= -30,NLangle= 20](0)(0,-5){$a$}
\end{gpicture}
\begin{gpicture}[name=HMSC-example-M2,ignore]
  \gasset{Nframe=n,Nfill=n,AHnb=0}
  \node(p)( 0,3){$p$} \drawline( 0,0)( 0,-12)
  \node(q)(10,3){$q$} \drawline(10,0)(10,-12)
  \node(M)(5,-15){$M_q$}
  \gasset{Nw=1.6,Nh=1.6,Nfill=y,ExtNL=y,NLdist=0.8,AHnb=1}
	\node[Nmarks=f,fangle=-150,NLangle=160](0)(10,-5){$b$}
\end{gpicture}
\begin{gpicture}[name=HMSC-example-M3,ignore]
  \gasset{Nframe=n,Nfill=n,AHnb=0}
  \node(p)( 0,3){$p$} \drawline(  0,0)(  0,-12)
  \node(q)(10,3){$q$} \drawline(10,0)(10,-12)
  \node(M)(5,-15){$M_p'$}
  \gasset{Nw=1.6,Nh=1.6,Nfill=y,ExtNL=y,NLdist=0.8,AHnb=1}
  	\node[Nmarks=i,iangle=150,NLangle=-170](0)(10, -5){$a$}
\end{gpicture}
\begin{gpicture}[name=HMSC-example-M4,ignore]
  \gasset{Nframe=n,Nfill=n,AHnb=0}
  \node(p)( 0,3){$p$} \drawline(  0,0)(  0,-12)
  \node(q)(10,3){$q$} \drawline(10,0)(10,-12)
  \node(M)(5,-15){$M_q'$}
  \gasset{Nw=1.6,Nh=1.6,Nfill=y,ExtNL=y,NLdist=0.8,AHnb=1}
	\node[Nmarks=i,iangle= 30,NLangle=-10](0)(0, -5){$b$}
\end{gpicture}
\begin{gpicture}[name=HMSC-example,ignore]
  \gasset{Nw=6,Nh=6,loopdiam=6,loopwidth=5}
  \node[Nmarks=i](1)(0,0){1}
  \node(2)(20,0){2}
  \node(3)(40,0){3}
  \node(4)(60,0){4}
  \node[Nmarks=f](5)(80,0){5}
  \drawloop(1){$M_p$}
  \drawedge(1,2){$M_p$}
  \drawloop(2){$M_q$}
  \drawedge(2,3){$M_q$}
  \drawloop(3){$M_p'$}
  \drawedge(3,4){$M_p'$}
  \drawloop(4){$M_q'$}
  \drawedge(4,5){$M_q'$}
\end{gpicture}

\begin{gpicture}[name=CFM-example,ignore]
  \gasset{Nw=6,Nh=6,loopdiam=6,loopwidth=5}
  \node[Nmarks=i](1)(0,0){1}
  \node[Nmarks=f](2)(20,0){2}
  \drawloop(1){$p!q, a$}
  \drawedge(1,2){$p!q, a$}
  \drawloop(2){$p?q, b$}
    
  \node[Nmarks=i](1)(60,0){1}
  \node[Nmarks=f](2)(80,0){2}
  \drawloop(1){$q!p, b$}
  \drawedge(1,2){$q!p, b$}
  \drawloop(2){$q?p, a$}
  
  \gasset{Nframe=n,Nfill=n,AHnb=0}
  \node(Ap)(-12,0){$\CFM_p$:}
  \node(Aq)(48,0){$\CFM_q$:}
\end{gpicture}

\vspace{-14ex}

\section{Introduction}\label{sec:introduction}

Message Sequence Charts (MSCs) provide a visual formalism for representing interactions between processes in concurrent or distributed systems that communicate through FIFO channels. An MSC illustrates a potential execution by recording send, receive, and local events along the relative timelines of each process, with each pair of matching send and receive events connected by an arrow. This intuitive representation makes MSCs appealing for specifying the behavior of such systems.

Specification languages that extend MSCs, such as high-level message sequence charts
(HMSCs), also known as message sequence graphs (MSGs), have been extensively studied.
HMSCs enrich MSCs by introducing constructs like choice, concatenation, and iteration,
similar to automata or regular expressions.
 This extension allows for specifying a set of behaviors, whether desired
(positive specification) or prohibited (negative specification), while maintaining a
global view of the system.
Despite the high-level, declarative view provided by HMSCs, actual system implementations
operate at a lower, more granular level.  Automata models, mainly communicating
finite-state machines (CFMs), are used to abstract these implementations.  In these
models, transitions are no longer labeled by MSCs but by atomic actions such as send and
receive events.  The implementation must ensure that these atomic actions align correctly
to meet the global specifications.

A significant amount of research has studied the relationships between different
specification languages and their translations into CFMs, e.g.,\cite{Morin02,Genest05,GenestKM06,HenriksenMKST05,Kuske03,DR1995,BolligFG21,GenestMSZ06,AlurEY03,AlurEY05,Lohrey03}. While translating specifications
directly into CFMs is referred to as the \emph{synthesis problem},
we will refer to the existence of an implementation as \emph{realizability}.
Realizability/synthesis are particularly valuable because they ensure implementations that are correct by design.
Early results focused on HMSCs as the specification language and communicating finite-state
machines as the implementation model, under the constraint that the automata could not
carry additional message contents beyond those specified in the HMSC
\cite{AlurEY03,AlurEY05,Lohrey03}.

When CFMs are allowed to include additional message contents, they become more expressive, leading to results analogous to Büchi and Kleene theorems.\footnote{Another variant of the realizability problem permits implementations to include additional communication, allowing for send and receive events not specified initially~\cite{Genest05}.} These results establish an expressive equivalence between HMSCs, monadic second-order logic, and CFMs, assuming bounded channels. There is, however, an important distinction in the type of channel bounds considered: universally bounded channels require that no execution exceeds a specified bound \cite{HenriksenMKST05,Kuske03}. In contrast, existentially bounded channels allow for executions with an equivalent bounded one, meaning they can be reordered to fit within the bound \cite{GenestKM06}. Note that these results have interesting connections to the theory of star-connected rational expressions over Mazurkiewicz traces and their implementations via Zielonka automata in shared-memory systems \cite{Kuske03,GenestKM06,DR1995}. Moreover, they have been used as a framework to establish corresponding results in the realm of multiparty session types \cite{HYC16,SZ22}. While research on HMSCs has primarily addressed MSC languages with bounded channels, it was shown in \cite{BolligFG21} that any first-order definable property, giving access to the ``happen-before'' relation and without assuming any channel capacity, can be implemented as a CFM.

In this paper, we identify a class of HMSCs that define languages allowing for unbounded
channels, extending beyond existentially bounded channels while guaranteeing effective
translation into a CFM implementation.  Similarly to globally cooperative
HMSCs~\cite{GenestKM06,GenestMSZ06}, this class is characterized by a restriction on the
iteration to sets of connected MSCs, inspired by Mazurkiewicz trace theory.  Intuitively,
this restriction facilitates process communication and prevents the global HMSC
specification from enforcing patterns, like an equal number of messages, between groups of
processes that do not communicate sufficiently.  However, our approach introduces a novel
scope for iteration, relying on a graph-based view of MSC connectedness (nodes are events 
of the MSC and edges are induced by process successor and matched send-receive pairs) rather than a
traditional communication graph (nodes are processes and edges are induced by possibly 
unmatched send-receive events).  To establish our main result, we translate HMSCs into
existential monadic second-order logic and leverage \cite[Theorem~4]{BolligFG21}.

It is important to note that, like in \cite{GenestKM06} over existentially bounded channels, the implementations derived from our approach are inherently non-deterministic. Deterministic machines can only be obtained under universally bounded channels \cite{GenestKM07,HenriksenMKST05}. Moreover, for deadlock-free implementations, additional constraints are required (e.g., \cite{GenestMSZ06, BaudruM07}).

Besides the realizability problem, i.e., the problem of translating HMSCs into CFMs, we also address another fundamental problem: \emph{satisfiability}. Satisfiability consists of asking whether a given HMSC defines \emph{some} behavior at all. It is akin to reachability or nonemptiness problems (e.g., in Petri nets or automata). However, it turns out that these problems are undecidable for unbounded channels already under strong restrictions on HMSCs. This also underpins the importance of positive findings concerning the realizability problem.

\paragraph{Outline.}

Section~\ref{sec:prel} introduces the fundamental concepts, including (high-level) message
sequence charts (HMSCs), communicating finite-state machines, and (existential)
monadic second-order logic.
Section~\ref{sec:main-results} states our main results for both realizability and
satisfiability and highlights how they compare with previous work. The detailed proofs are
presented in subsequent sections. The positive realizability result is
obtained in several steps.
In Section~\ref{sec:closure}, we
provide closure properties of MSC languages definable in existential monadic second-order
logic. In Section~\ref{sec:realizability}, these closure properties are exploited
for an inductive translation of HMSCs into CFMs.
The undecidability proofs concerning satisfiability are presented in Section~\ref{sec:satisfiability}.
Finally, Section~\ref{sec:conclusion} discusses
open problems and potential directions for future research.

\section{Preliminary Definitions}\label{sec:prel}

We fix a nonempty finite set of \emph{processes} $\Procs$.
We assume P2P communication through \emph{channels} from
$\Ch = \{(p,q) \in \Procs \times \Procs \mid p \neq q\}$.
The set of \emph{actions} (or \emph{action types}) of process $p \in \Procs$ is
$\Act_p = \Send_p \cup \Rec_p$ where
$\Send_p = \{ \send{p}{q} \mid (p,q) \in \Ch\}$
is the set of send actions and
$\Rec_p = \{ \rec{p}{q} \mid (q,p) \in \Ch\}$ is the set of receive actions.
Here, $\send{p}{q}$ ($\rec{p}{q}$, respectively)
denotes that process $p$ sends a message to (receives a message from, respectively)
process $q$.
We let $\Send = \bigcup_{p \in \Procs} \Send_p$ and
$\Rec = \bigcup_{p \in \Procs} \Rec_p$. Finally,
$\Act = \Send \cup \Rec$ is the set of all actions.

\subsection{Message Sequence Charts}

The atomic building blocks of HMSCs are MSCs. An MSC denotes a fragment of a communication scenario. As a fragment, it can be \emph{compositional}~\cite{GunterMP03} (also referred to as partial) in the sense that a send or receive event is not necessarily matched by a corresponding communication event. However, such an event may later be matched when combined with other compositional MSCs within an HMSC.
The following definition is illustrated by Example~\ref{ex:msc-example} below.

\begin{definition}[compositional MSC]\label{def:cmsc}
  Let $\Msg$ be a finite set of \emph{messages}.
  A \emph{compositional MSC (cMSC)} over $\Procs$ and $\Msg$
  is a tuple $\msc = (E,\le,\msgrel,\lab, \mu)$.  Here, $E$ is the nonempty, finite or countably infinite,
  set of \emph{events}, which is equipped with a partial order ${\le} \subseteq E \times
  E$.  Furthermore, $\lab\colon E \to \Act$ is a labeling
  function, and $\mu\colon E \to \Msg$ is a partial function.

  For $p \in \Procs$, $E_p = \{e \in E \mid \lab(e) \in \Act_p\}$ is the
  set of events executed by process $p$.  We denote the restriction of $\le$ to $E_p$ by ${\le_p}
  \subseteq E_p \times E_p$ and require that $(E_p,\le_p)$ be a total order that is finite
  or isomorphic to $(\mathbb{N},\leq)$.  By ${\procrel_p}$, we denote the direct-successor
  relation of $\le_p$, and we let ${\procrel}=\bigcup_{p\in\Procs}{\procrel}_{p}$.
  The relation ${\msgrel}\subseteq {\le}$ matches send and receive
  events according to a first-in-first-out (FIFO) policy: (i) for all $(e,f) \in {\msgrel}$, there is $(p,q)
  \in \Ch$ such that $\lab(e) = \send{p}{q}$ and $\lab(f) = \rec{q}{p}$; (ii) for all
  $(e,f),(e',f') \in {\msgrel}$ and $(p,q) \in \Ch$ such that $e,e' \in E_p$ and $f,f' \in
  E_q$, we have $e \le_p e'$ iff $f \le_q f'$.  The partial order $\le$ has to be equal to
  the reflexive transitive closure of ${\procrel}\cup{\msgrel}$.

  Note that there may be unmatched send and receive events, which we gather in the set
  $\unm{\msc} = \{e \in E \mid$ there is no $f \in E$ such that $e \msgrel f$ or $f \msgrel e\}$.
  Messages from $\Msg$ are used to label unmatched send and receive events\footnote{The use of messages from $\Msg$ to identify suitable communication events is an extension with respect to previous work. It increases the expressive power of HMSCs but does not introduce considerable additional technical complexity.}, which
  will allow us to selectively match send and receive events when concatenating several
  cMSCs.  We require that $\dom(\mu)=\unm{\msc}$.
  
  Finally, we require that, for all $(e,f) \in {\msgrel}$ and $g\in\unm{\msc}$,
  $\lab(g)=\lab(e)$ implies $e<g$, and $\lab(g)=\lab(f)$ implies $g<f$
  (which avoids unmatched events that cannot be	matched later on).  \qed
\end{definition}

We call a cMSC $\msc=(E,\le,\msgrel,\lab,\mu)$ an MSC if $\unm{\msc} = \emptyset$. We may then omit $\mu$ and just write $(E,\le,\msgrel,\lab)$.
We call $\msc$ \emph{finite} if $\mEvents$ is finite.
By $\pMSCs$, $\fcMSCs$, $\MSCs$ and $\fMSCs$, we denote the sets of compositional MSCs, finite
compositional MSCs, MSCs, and finite MSCs, respectively.  

\begin{example}\label{ex:msc-example}
Figure~\ref{fig:msc-examples} depicts six finite cMSCs over $\Procs = \{p, q\}$ and
$\Msg = \{a\}$. cMSC $M_1$ has the three events $e$, $e'$, and $f$. We have, e.g.,
$E_p = \{e, e'\}$, $e \le_p e'$, $e \procrel e'$,
$e \msgrel f$, $\lab(e) = \lab(e') = p!q$, and $\mu(e') = a$. Moreover,
$\unm{\msc_1} = \{e'\}$.
\end{example}

\begin{figure}[tb]
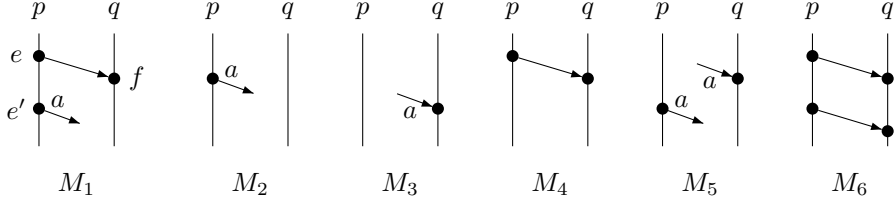

  \centering
  \gusepicture{cMSC-example-M1}\hspace{2em}
  \gusepicture{cMSC-example-M2}\hspace{2em}
  \gusepicture{cMSC-example-M3}\hspace{2em}
  \gusepicture{cMSC-example-M4}\hspace{2em}
  \gusepicture{cMSC-example-M5}\hspace{2em}
  \gusepicture{cMSC-example-M6}
  \caption{Example cMSCs.}
  \label{fig:msc-examples}
\end{figure}

\begin{remark}
All results presented in this paper hold when extending the definitions in two respects: First, we can add finitely many \emph{attributes}, where an attribute may provide additional information about an event. Second, we can include another type of events, \emph{internal} events, which are exclusive to a process and neither emit nor receive messages (such as ``enters critical region'').
\end{remark}

\paragraph{Concatenation.}
We define concatenation of cMSCs following \cite{GenestKM06} but taking into account messages.
For cMSCs $\msc_1=(E^1,\le^1,\msgrel^1,\lab^1,\mu^1) \in \fcMSCs$ and $\msc_2=(E^2,\le^2,\msgrel^2,\lab^2,\mu^2) \in \pMSCs$,
the \emph{concatenation} $\msc_1 \circ \msc_2 \subseteq \cMSCs$
is defined as the \emph{set} of cMSCs $\msc = (E,\le,\msgrel,\lab)$
such that the following hold:
\begin{itemize}[nosep]\itemsep=0.5ex
  \item $E=E^1 \uplus E^2$,
  
  \item ${\le_p} = {\le_p^1} \cup {\le_p^2} \cup (E_p^1 \times E_p^2)$ for all $p \in \Procs$,

  \item $(\msgrel^1 \cup \msgrel^2) \subseteq {\msgrel}$ and for all $(e,f) \in \msgrel
  \setminus (\msgrel^1 \cup \msgrel^2)$, we have $e \in E^1$, $f \in E^2$, and 
  $\mu^1(e) = \mu^2(f)$,
  
  \item for all $e\in E^i$ ($i\in\{1,2\}$), we have $\lab(e)=\lab^{i}(e)$ and,
  if $e\in\unm{\msc}$, then $\mu(e)=\mu^{i}(e)$.
\end{itemize}
\smallskip

In other words, the concatenation of two cMSCs is their ``vertical stacking'', where one is written below the other while possibly matching send events from the first with receive events from the second.

\begin{example}
Consider once again Figure~\ref{fig:msc-examples}.
The concatenation $M_2 \circ M_3$ equals $\{M_4, M_5\}$.
In $M_4$, a message edge has been added from the send to the receive event (note that synchronization message $a$ disappears), while in $M_5$, the edge has not been added.
On the other hand, $M_3 \circ M_2$ is the singleton set $\{M_5\}$ (a send event can only be matched with a receive event from a \emph{subsequent} cMSC). 
The cMSC $M_6$ is contained in the concatenation $M_1 \circ M_3$. In fact, $M_6$ is the only cMSC in $M_1 \circ M_3$: leaving the open receive event in $M_3$ unmatched would result in a structure where this event can never be matched via a concatenation, and which is, therefore, excluded according to Definition~\ref{def:cmsc} (cf.\ the very last condition).
\end{example}

We extend concatenation to sets $L_1 \subseteq \fcMSCs$ and $L_2 \subseteq \cMSCs$
via
$L_1 \circ L_2 = \bigcup_{M_1 \in L_1, M_2 \in L_2} M_1 \circ M_2$.
Abusing notation, we abbreviate $\{M\} \circ L$ and $L \circ \{M\}$
by $M \circ L$ and $L \circ M$, respectively. Note that concatenation is associative.

\paragraph{Infinite product.}
We extend concatenation to infinite sequences.  For all $i\geq1$, let
$\msc_i=(E^i,\le^i,\msgrel^i,\lab^i,\mu^i)\in\fcMSCs$ be \emph{finite} cMSCs.
The \emph{product} $\prod_{i\geq1}\msc_i \subseteq \cMSCs$ is defined as the \emph{set} of
cMSCs $\msc = (E,\le,\msgrel,\lab,\mu)$ such that the following hold:
\begin{itemize}[nosep]\itemsep=0.5ex 
  \item $E=\biguplus_{i\geq1}E^i$,
  
  \item ${\le_p} = \big( \bigcup_{i\geq1}{\le_p^i} \big) \cup
  \big( \bigcup_{1\leq i<j} (E_p^i \times E_p^j) \big)$
  for all $p \in \Procs$,

  \item $\big( \bigcup_{i\geq1}\msgrel^i \big) \subseteq {\msgrel}$ and for all $(e,f) \in
  \msgrel \setminus \big( \bigcup_{i\geq1}\msgrel^i \big)$, we have $e \in E^i$, $f \in
  E^j$ for some $1\leq i<j$, and $\mu^i(e) = \mu^j(f)$,
  
  \item for all $e\in E^i$ ($i\geq1$), we have $\lab(e)=\lab^{i}(e)$ and,
  if $e\in\unm{\msc}$, then $\mu(e)=\mu^{i}(e)$.
\end{itemize}
The infinite product is extended to languages $(L_i)_{i\geq1}\subseteq\fcMSCs$ as 
expected. The infinite product is also associative.

\subsection{High-Level Message Sequence Charts}

To obtain a convenient specification language, cMSCs are combined
towards HMSCs using constructs such as choice, concatenation, and iteration.
This gives rise to the notion of (compositional) high-level message sequence charts~\cite{GunterMP03}.

\begin{definition}
A \emph{high-level message sequence chart} (HMSC) over $\Procs$
is a tuple $\HMSC = (S, \iota, \Msg, R, F, F_\omega)$ where
$S$ is a finite set of states, $\iota$ is the initial state,
$\Msg$ is a finite set of messages,
$R \subseteq S \times \fcMSCs \times S$ is the finite transition relation
(recall that $\fcMSCs$ is the set of finite cMSC over $\Procs$ and $\Msg$),
and $F, F_\omega \subseteq S$ are sets of accepting states (one for finite paths, and one for infinite paths).
\end{definition}

A finite path in $\HMSC$ is a sequence
$\rho = s_0 \msc_1 s_1 \ldots \msc_n s_n \in S(\fcMSCs\, S)^\ast$
with $n \ge 1$, such that, for all $i \in \{1,\ldots,n\}$,
we have $(s_{i-1},\msc_i,s_i) \in R$. We say that $\rho$ is a path from
$s_0$ to $s_n$. Infinite paths are defined similarly.
A path is accepting if it starts in the initial state and,
(i) if it is finite and ends in an accepting state from $F$, (ii)
if it is infinite and visits an accepting state from $F_\omega$
infinitely often. The concatenation of all cMSCs in a path $\rho$
is denoted $\mscrun{\rho}$. Recall that $\mscrun{\rho}$ is a \emph{set} of 
cMSCs.
The language of $\HMSC$ is the set 
$L(\HMSC) = \bigcup \{\mscrun{\rho} \mid \rho \textup{ is an accepting path of } \HMSC\}
\cap \MSCs$. 
Note that $L(\HMSC)$ contains only MSCs, rather than all cMSCs.

\begin{example}
Figure~\ref{fig:hmsc-example} illustrates an HMSC $\HMSC$ over $\Procs = \{p, q\}$.
In particular, we have $\Msg = \{a, b\}$ and $F = \{5\}$. The figure also depicts
an MSC $M \in L(\HMSC)$.
\end{example}

It should be noted that, compared to previous definitions, adding messages to HMSCs in terms of $\Msg$ increases their expressive power when not restricting to channel-bounded languages. Messages in HMSCs are akin to messages in communicating finite-state machines, as defined below. In fact, with the use of messages, every communicating finite-state machine can be translated into an equivalent HMSC (though we do not delve into the details in this paper),
while this is in general not possible without messages \cite[Proposition~5.5.1]{Bollig2005}.
Moreover, it is easy to conceive an HMSC that generates a language that is non-regular in a certain sense. Consider a single-state HMSC with $\Procs = \{p_1, p_2, p_3, p_4 \}$, looping over an MSC containing two complete message edges: one from process $p_1$ to process $p_2$, and another from process $p_3$ to process $p_4$. The resulting language consists of MSCs in which the number of messages from $p_1$ to $p_2$ matches the number of messages from $p_3$ to $p_4$. This language is not realizable by communicating finite-state machines\footnote{In fact, it is not even definable in monadic second-order logic.}. Since there is no communication between the groups of processes 
$\{p_1, p_2\}$ and $\{p_3, p_4\}$, any asynchronous implementation of this language would require a counting mechanism to compare the unbounded number of messages.

To define ``feasible'' classes of HMSCs, we restrict loops to connected cMSCs.
We call a cMSC $\msc = (E,\le,\msgrel,\lab,\mu)$ \emph{connected} if the undirected graph
$(\mEvents,{\le} \cup {\le}^{-1})$ is connected.

\begin{definition}
An HMSC is called \emph{loop-connected} if, for all $s \in S$ and all 
finite
paths $\rho$ from $s$ to $s$,
all cMSCs in $\mscrun{\rho}$ are connected.
\end{definition}

\begin{example}
Every cMSC in Figure~\ref{fig:msc-examples} is connected, except $M_5$.
  The HMSC depicted in Figure~\ref{fig:hmsc-example} is loop-connected:
  All loops generate cMSCs where all events are on $p$ or all events are on $q$.
Any such cMSC is connected.
\end{example}

\begin{figure}[tb]
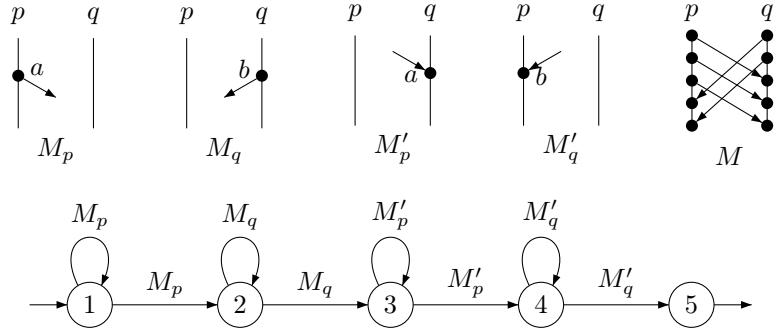

  \centering
  \gusepicture{HMSC-example-M1}\hfil
  \gusepicture{HMSC-example-M2}\hfil
  \gusepicture{HMSC-example-M3}\hfil
  \gusepicture{HMSC-example-M4}\hfil
  \gusepicture{HMSC-example-M}
  \\[3ex]
  \gusepicture{HMSC-example}
  \caption{Example HMSC that is loop-connected but not safe}
  \label{fig:hmsc-example}
\end{figure}

\subsection{Communicating Finite-State Machines}

Communicating finite-state machines (CFMs) are the operational counterpart to HMSCs
and represent a low-level model of distributed systems. While HMSCs offer the possibility
to connect specific send and receive events, synchronization in CFMs is solely accomplished
via FIFO channels. 

\newcommand{\Acc}{\mathit{Acc}}
\newcommand{\Com}{\mathit{Com}}
\newcommand{\cfmrun}{\tau}
\newcommand{\staterun}{\sigma}
\newcommand{\msgrun}{\mu}
\newcommand{\act}{a}
\newcommand{\trans}{\theta}
\newcommand{\tsource}{\mathit{source}}
\newcommand{\ttarget}{\mathit{target}}
\newcommand{\taction}{\mathit{action}}
\newcommand{\tmsg}{\mathit{message}}
\newcommand{\translabel}{a}

We give here an informal account of CFMs, as they are not needed for
our technical developments. A formal definition can be found in \cite{BolligFG21}.
Recall that we fixed the set $\Procs$ of processes.
In a CFM, every process $p \in\Procs$ is represented by a finite-state machine $\CFM_p$,
which allows it to execute actions from $\Act_p$ and to send and receive messages from
a set $\Msg$. Messages in transit from process $p$ to process $q$ are stored in the
(unbounded) FIFO channel $(p,q)$. However, like for HMSCs, they do not occur in (complete) MSCs.
 In addition, we have simple global acceptance conditions, which take into account
finite and/or infinite behaviors. The language of a CFM $\CFM$ is denoted by $L(\CFM)$.

\begin{example}
Figure~\ref{fig:cfm-example} shows a CFM $\CFM$ over $\Procs = \{p, q\}$ that is equivalent to the HMSC
$\HMSC$ from Figure~\ref{fig:hmsc-example}, i.e., $L(\CFM) = L(\HMSC)$.
\end{example}

\begin{figure}[tb]
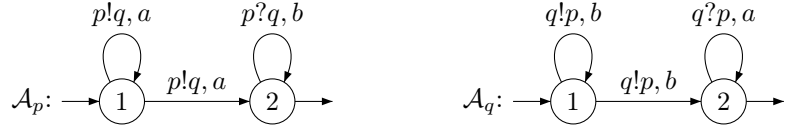

  \centering
  \gusepicture{CFM-example}
  \caption{An example CFM}
  \label{fig:cfm-example}
\end{figure}

\subsection{(Existential) Monadic Second-Order Logic}

It has been shown in \cite{BolligFG21} that CFMs are expressively equivalent to
existential monadic second-order logic (EMSO).  EMSO formulas over
given sets $\Procs$ of processes and $\Msg$ of messages are interpreted
over cMSCs over $\Procs$ and $\Msg$. They are of the form $\phi =
\exists X_1 \ldots \exists X_k.  \psi$.  Here, the $X_i$ are existentially quantified
second-order variables, interpreted as sets of events of a given cMSC.
Moreover, $\psi$ is a
first-order formula that (i) can make use of boolean connectives (disjunction,
conjunction, and negation), (ii) has access to the process order, direct process-successor relation,
and message relation in
terms of formulas $x \le_p y$, $x \procrel y$, and $x \msgrel y$ ($x,y$ being first-order variables
interpreted as events), (iii) can determine the type or message of an event in terms of formulas
$p!q(x)$, $p?q(x)$, and $a(x)$ with $a \in \Msg$ (which is satisfied if the event
representing $x$ is unmatched and carries message $a$), (iv) uses first-order quantification $\exists x$ and $\forall x$,
and (v) can relate first-order and second-order variables using atomic formulas of the
form $x \in X_i$.

If $\phi$ is a sentence, i.e., no variable is free in $\phi$, we let $L(\phi)$ denote
the set of MSCs $M \in \MSCs$ that satisfy $\phi$ (written $M \models \phi$).

\begin{example}
Below is an example EMSO sentence $\phi$ over $\Procs = \{p, q\}$
(and no messages)
such that $L(\phi) = L(\HMSC) = L(\CFM)$,
where $\HMSC$ is the HMSC from Figure~\ref{fig:hmsc-example}
and $\CFM$ is the CFM from Figure~\ref{fig:cfm-example}:
\[
\phi =
\left(
\begin{array}{rl}
& \exists x, y.\bigl(p?q(x) \wedge \max(x) \wedge q?p(y) \wedge \max(y)\bigr)\\[1ex]
\wedge & \neg \exists x, y. \bigl( p?q(x) \wedge p!q(y) \wedge x \le_p y \bigr)\\[1ex]
\wedge & \neg \exists x, y. \bigl( q?p(x) \wedge q!p(y) \wedge x \le_q y \bigr)
\end{array}
\right)
\]
Here, the abbreviation $\max(z) = \neg \exists z'. z \procrel z'$ states that event $z$ is maximal on its process.
Thus, $\phi$ says that a given MSC
\begin{itemize}
\item is finite,
\item has at least one event of the form $p?q$ (thus, at least one of the form $q!p$),
\item has at least one event of the form $q?p$ (thus, at least one of the form $p!q$),
\item on both processes, all send events are scheduled before all receive events.
\end{itemize}
\end{example}

\begin{theorem}[\!\!\cite{BolligFG21}]\label{thm-EMSO-CFM}
  For every EMSO sentence $\phi$, one can effectively construct
  a CFM $\CFM$ such that $L(\phi) = L(\CFM)$.
\end{theorem}

\section{Main Results}
\label{sec:main-results}

We consider HMSCs as a specification language and CFMs
as a model of implementation. Accordingly, we address two types of questions:
\emph{satisfiability} of an HMSC (does a given HMSC produce any behavior?) and
\emph{realizability} of an HMSC language as a CFM (sometimes also referred to as
\emph{implementability}). While satisfiability is formulated as a decision problem,
realizability aims to identify classes of HMSCs that \emph{guarantee} the existence
of a corresponding CFM. In this section, we summarize our results
for both problems. The proofs are then spread over subsequent sections.

\subsection{Realizability of HMSCs}

Realizability results identify classes of HMSCs
that allow for an effective translation into equivalent CFMs.
In \cite{GenestKM06}, it was shown that all \emph{globally cooperative}
HMSCs are realizable in that sense.
While globally cooperative HMSCs specify existentially bounded behaviors
(where each MSC has some linearized execution that does not exceed a given channel capacity \cite{GenestKM06}),
our main result addresses realizability for the class of loop-connected HMSCs, which
allows for unbounded behaviors (like the language generated by the HMSC from Figure~\ref{fig:hmsc-example}).

\begin{theorem}
\label{thm:HMSC-to-CFM}
  Given a loop-connected HMSC $\HMSC$, we can effectively
  construct a CFM $\CFM$ such that $L(\HMSC) = L(\CFM)$.
\end{theorem}

The proof of Theorem~\ref{thm:HMSC-to-CFM} relies on Theorem~\ref{thm-EMSO-CFM}:
we only need to prove that any loop-connected HSMC can be translated into
an equivalent EMSO formula.
We first show in Section~\ref{sec:closure} that EMSO-definable languages
of cMSCs are closed under union and concatenation, and in the case where
all cMSCs in the language are connected, under iteration.
We then use these closure properties together with standard state-elimination 
techniques in Section~\ref{sec:realizability} to construct, from a 
loop-connected HMSC $\HMSC$, an EMSO formula $\phi$ such that
$L(\phi) = L(\HMSC)$.

\paragraph{The case of bounded channels.}

We will now compare Theorem~\ref{thm:HMSC-to-CFM} to \cite{GenestKM06},
which shows an analogous result for \emph{globally cooperative} HMSCs over finite cMSCs.
Globally cooperative HMSCs are based on the notion of \emph{weakly connected} cMSCs.\footnote{The property of being weakly loop-connected is referred to as \emph{loop-connected} in \cite{GenestKM06}.}

A cMSC $\msc = (E,\le,\msgrel,\lab,\mu)$ is
\emph{weakly connected} if it has a connected (undirected) communication graph.
Here, the communication graph of $\msc$ has $\{p \in \Procs \mid E_p \neq \emptyset\}$ as
set of nodes and it has an (undirected) edge between $p$ and $q$ if, for some $e,f \in E$,
we have
$\lab(e) = \send{p}{q}$ and $\lab(f) = \rec{q}{p}$ or
$\lab(e) = \send{q}{p}$ and $\lab(f) = \rec{p}{q}$.
Note that \emph{connected} implies \emph{weakly connected}, but not the other way around.
For example, the cMSC $M_5$ in Figure~\ref{fig:msc-examples}
is weakly connected, but not connected.

Analogously to the definition of loop-connected,
an HMSC is called \emph{weakly loop-connected} if, for all $s \in S$ and all 
finite paths $\rho$ from $s$ to $s$,
all cMSCs in $\mscrun{\rho}$ are weakly connected.

Let $\HMSC = (S, \iota, \Msg, R, F, F_\omega)$ be an HMSC such that
$F_\omega = \emptyset$.
We call $\HMSC$ \emph{safe}
if, for all accepting paths $\rho$ of $\HMSC$, $\mscrun{\rho}$ contains an MSC \cite{GenestKM06}.
We call $\HMSC$ \emph{globally cooperative}
if it is weakly loop-connected and safe.\footnote{The definition can be extended to
include infinite cMSCs, but it is more technical.}

\begin{example}
  The HMSC depicted in Figure~\ref{fig:hmsc-example} is not safe and, therefore, not
  globally cooperative.  In particular, there is no uniform bound on the channel capacity
  that would allow all behaviors to be scheduled in a way such that no channel exceeds the
  bound.  That is, there is no equivalent HMSC that is globally cooperative.
\end{example}

\begin{theorem}[Genest, Kuske, and Muscholl \cite{GenestKM06}]
\label{thm:gcHMSC-to-CFM}
  Every globally cooperative HMSC $\HMSC$ (with $F_\omega = \emptyset$)
  can be effectively translated into a CFM $\CFM$ such that $L(\HMSC) = L(\CFM)$.
\end{theorem}

We leave it as an open problem whether, for every globally cooperative HMSC, there is an
equivalent loop-connected HMSC.

\subsection{Satisfiability of HMSCs}

Deciding whether the language of a given HMSC is nonempty turns out to be undecidable,
even under the structural restriction of loop-connected HMSCs.

\begin{theorem}\label{thm:undecidability}
  The following decision problem is undecidable: Given a loop-con\-nect\-ed HMSC $\HMSC$ (with at least
  two processes), do we have $L(\HMSC) \neq \emptyset$?
\end{theorem}

We give two proofs of this result: the first one, in Section~\ref{sec:sat-flat},
goes through a reduction from Post correspondence problem, and shows that 
satisfiability is undecidable even for \emph{flat} HMSCs, but uses
three processes. The second proof, in Section~\ref{sec:sat-2proc} is a 
reduction to the halting problem, and proves that satisfiability is undecidable
even with \emph{two} processes -- but it involves a non-flat HSMC.

\medskip

We also show in Section~\ref{sec:sat-sgm} that restricting to HMSCs where $\Msg$ is a singleton set results in an
undecidable satisfiability problem (though here we do not assume that the HMSC is loop-connected).
Note that this corresponds to the standard definition of HMSCs (cf., for example, \cite{GenestKM06}).

\begin{theorem}\label{thm:singleton}
  The following decision problem is undecidable: Given an HMSC
  $\HMSC$ such that $\Msg$ is a singleton set,
  do we have $L(\HMSC) \neq \emptyset$?
\end{theorem}

\section{Some Closure Properties of EMSO Languages}\label{sec:closure}

In this section, we fix a finite set $\Msg=\{m_1,\ldots,m_\ell\}$ of messages
and we let $n=|\Procs|$ be the number of processes.
We consider (languages of) cMSCs over $\Procs$
and $\Msg$.

\begin{lemma}\label{lem:emso-union-concatenation}
  EMSO languages of cMSCs are closed under union and concatenation.
\end{lemma}

\begin{proof}
  Let $L_{1},L_{2}$ be languages of cMSCs defined by EMSO sentences 
  $\Phi_{1},\Phi_{2}$. Without loss of generality, we may assume that 
  $\Phi_{1},\Phi_{2}$ use the same \emph{set} variables: 
  $\Phi_{1}=\exists X_1 \ldots \exists X_k. \varphi_{1}$ and
  $\Phi_{2}=\exists X_1 \ldots \exists X_k. \varphi_{2}$ where $\varphi_{1},\varphi_{2}$ are 
  first-order formulas.
  
  Clearly, the language $L_{1}\cup L_{2}$ is defined by the formula
  $\exists X_1 \ldots \exists X_k. (\varphi_{1}\vee\varphi_{2})$.
  
  We show that the concatenation $L_{1}\circ L_{2}$ is defined by a 
  sentence $\Psi$ of the form
  $$
  \Psi = \exists W.\, 
  \exists Y_{m_1} \ldots \exists Y_{m_\ell}.\, 
  \exists X_1 \ldots \exists X_k.\, \psi \, .  
  $$
  The intuition is that variable $W$ identifies the prefix of the composed cMSC and
  variables $Y_{m_1}, \ldots, Y_{m_\ell}$ are used to guess which messages are "new" in
  the composition, and what the original message labels were.
  The first-order part $\psi$ of $\phi$ is a conjunction
  $\psi_{1}\wedge\psi_{2}\wedge\psi_{3}\wedge\psi_{4}$ checking several conditions.  The
  first requirement in $\psi$ makes sure that $W$ identifies a 
  \emph{finite} (nonempty) prefix of the cMSC:
  \begin{align*}
    \psi_{1} &= \exists x_{1}\ldots \exists x_{n} \forall y.\, y\in W 
    \Longleftrightarrow \bigvee_{1\leq i\leq n} y\leq x_{i} \,.
  \end{align*}

  The second condition says that (only) newly matched send events in $W$ and receive
  events not in $W$ can be in some $Y_m$, and the message must be unique.
  \begin{align*}
    \psi_{2} ={} &\forall x \forall y.\, x\msgrel y \implies 
    \Big(x,y\in W \vee x,y\notin W \vee \bigvee_{m\in\Msg} x,y\in Y_{m} \Big) 
    \\
    &\wedge\bigwedge_{m \in \Msg} \forall x.\, x \in Y_m \implies 
    \bigwedge_{m' \neq m} x \notin Y_{m'} \wedge \exists y.\, 
    \\
    &\hspace{20mm} (x \msgrel y \wedge x\in W \wedge y\notin W)
    \lor (y \msgrel x \wedge y\in W \wedge x\notin W)  
    \,.
  \end{align*}
  
  The third condition says that the prefix of the cMSC identified by $W$ satisfies
  $\Phi_{1}$.  To do so, we define a relativisation $\phiw \xi$ of a first-order
  formula $\xi$ to elements in $W$.  Message edges with endpoints in $Y_m$ should also be
  ignored, and the label $m$ added.  This is defined inductively as 
  follows: 
  \begin{align*}
    \phiw {(\exists y.\, \xi)} & = \exists y. \, y\in W \land \phiw \xi \\
    \phiw{m(y)} & = m(y) \lor y \in Y_m \qquad (m \in \Msg)
  \end{align*}
  The other cases are straightforward:
  \begin{align*}
    \phiw{(y \le z)} & = y \le z 
    &
    \phiw{(y \msgrel z)} &= y \msgrel z
    \\
    \phiw{p(y)} & = p(y) \qquad (p \in \Procs)
    &
    \phiw{(y\in X)} &= y\in X
    \\
    \phiw {(\xi_1 \lor \xi_2)} & = \phiw {\xi_1} \lor \phiw {\xi_2} 
    &
    \phiw {(\lnot \xi)} & = \lnot \phiw \xi \,.
  \end{align*}
  With this relativisation, the third formula is $\psi_{3}=\phiw{\varphi_{1}}$.  We define
  similarly the relativisation $\phinw{\xi}$ of a first-order formula $\xi$ to the
  suffix identified by the complement of $W$.  The last condition is 
  $\psi_{4}=\phinw{\varphi_{2}}$.
  \qed
\end{proof}

Let $L$ be a language of finite cMSCs. We say that $L$ is \emph{connected}
if all cMSCs in $L$ are connected.

\begin{lemma}\label{lem:emso-iteration-connected}
  Let $L$ be a language of \emph{finite} cMSCs which is definable in EMSO and
  \emph{connected}.  Then, $L^{+}$ and $L^\omega$ are definable in EMSO.
\end{lemma}

\begin{proof}
  Assume that $L$ is defined by the sentence 
  $\Phi = \exists X_1 \ldots \exists X_k.\, \phi$
  where $\phi$ is a first-order formula.
  We show that $L^{+}\cup L^{\omega}$ is defined by a sentence of the form
  \[
  \Psi =
  \exists W.\, 
  \exists Y_{m_1} \ldots \exists Y_{m_\ell}.\, 
  \exists X_1 \ldots \exists X_k.\, \psi \, .
  \]
  As in the proof of Lemma~\ref{lem:emso-union-concatenation}, variables $Y_{m_1}, \ldots,
  Y_{m_\ell}$ are used to guess which messages are "new" in the composition, and what the
  original message labels were. The main difference is with variable $W$ since, in an 
  iteration, we may have an arbitrary number of factors and not only a prefix and a suffix.
  The value of variable $W$ alternates on each process to identify factors of the
  composition.

  More precisely, consider a cMSC $\msc\in\msc_1\circ\msc_2\circ\cdots$ with $\msc_{i}\in
  L$ for all $i\geq 1$.  The interpretation of the set variables will be as follows:
  \begin{enumerate}[nosep]
    \item variables $X_1,\ldots,X_k$ are interpreted as the union of the interpretations
    witnessing the fact that $\msc_i\models\Phi$ for all $i$;

    \item variable $Y_m$ is interpreted as the set of send or receive events labeled with
    message $m$ in one of the $\msc_i$, and matched in $\msc$;
    
    \item variable $W$ is interpreted so that for all processes $p\in\Procs$, and all
    $i\geq 1$, either all $p$-events in $\msc_i$ are in $W$ or no $p$-events in
    $\msc_i$ are in $W$.  Moreover, if $\msc_i$ and $\msc_j$ ($i<j$) are consecutive
    factors with some $p$-events ($\msc_{k}$ has no $p$-events for $i<k<j$), then either
    all $p$-events in $\msc_i$ are in $W$ and no $p$-events in $\msc_j$ are in $W$, or the
    other way around.
  \end{enumerate}
  
  The first-order part $\psi$ of $\Psi$ is a conjunction 
  $\psi_{1}\wedge\psi_{2}\wedge\psi_{2}'\wedge\psi_{3}\wedge\psi_{4}$ checking several
  conditions.  In the following, we write
  $\Weq x y$ for $(x \in W \iff y \in W)$, and similarly 
  $\msgeq x y$ for the formula $\bigwedge_{m\in\Msg} (x \in Y_m \iff y \in Y_m)$.

  The first requirement in $\psi$ is a simple coherence condition on 
  $Y_{m_1},\ldots,Y_{m_\ell}$.
  Only matched send or receive events can be in some $Y_m$, and the message has
  to be unique.
  \[
  \psi_{1} =
  \bigwedge_{m \in \Msg} \forall x.\, x \in Y_m \implies \exists y.\, 
  (x \msgrel y \lor y \msgrel x) \land y \in Y_m \ \land
  \bigwedge_{m' \neq m} x \notin Y_{m'} \, .
  \]
  Notice that under the constraint $\psi_{1}$, a matched pair $e\msgrel f$ of send/receive
  events in $\msc$ are either both in the same $Y_m$ (the matching between them is added
  in the composition) or neither are in any $Y_m$ (they come from the same component
  $\msc_{i}$ and the matching was already there).

  \medskip

  To define the other conjuncts in $\psi$, we need to introduce a few more notations.  The
  main ingredient is a formula $x \sim y$ meant to be interpreted as $x$ and $y$ being
  part of the same factor in $L$.  This is where we use the fact that all cMSCs in
  $L$ are \emph{connected}.  Two events come from the same factor if ($\procsim$)
  they are on the same process and in the same $W$-block, or ($\msgsim$) they are related
  by a message which is not added in the composition (which can be checked using the
  $Y_{m}$'s), or they are in the transitive closure of these two basic 
  relations. We write ${\mprocle}=\bigcup_{p\in\Procs}{\leq}_{p}$ for the partial order 
  restricted to events on the same process.
  \begin{align*}
    x\procsim y &= (x\mprocle y \vee y\mprocle x) \land \forall z.
    \\
    &\qquad(x\mprocle z\mprocle y \vee y\mprocle z\mprocle x) \implies \Weq z x
    \\
    x\msgsim y &= (x\msgrel y \vee y\msgrel x) \wedge 
    \bigwedge_{m\in\Msg} \neg(x\in Y_{m} \vee y\in Y_{m})
    \\
    x\sim y &=
    \exists x_1,\ldots,x_{2n}.\, x=x_{1} \wedge y=x_{2n} \wedge
    \hspace{-1mm}\bigwedge_{1\leq i<2n}\hspace{-2mm}
    x_{i}\procsim x_{i+1} \vee x_{i}\msgsim x_{i+1} \,.
  \end{align*}
  It is easy to see that $\procsim$ is an \emph{equivalence} relation.  The relation
  $\sim$ is clearly \emph{reflexive} and \emph{symmetric}, but without further constraints,
  it need not be \emph{transitive} (notice that the path chosen between $x$ and $y$ is of
  length $2n$ where $n=|\Procs|$ is the number of processes).  To ensure transitivity, we
  request that, if two similar events $x\sim y$ are on the same process, then all events
  between $x$ and $y$ are in the same $W$-block, i.e., $x\procsim y$ (components are
  continuous on any given process).
  \begin{align*}
    \psi_{2} &=
    \forall x \forall y.\, (x\sim y \wedge x\mprocle y) \implies
    (\forall z.\, x\mprocle z\mprocle y \implies \Weq z x) \,.
  \end{align*}
  Under the constraint $\psi_{2}$, we can check that $\sim$ is transitive, hence it is an 
  equivalence relation.
  Indeed, assume that $x\sim y\sim z$ and let $x=x_{1},\ldots,x_{2n}=y,\ldots,x_{4n-1}=z$ 
  be a witnessing path, i.e., $x_{i}\procsim x_{i+1} \vee x_{i}\msgsim x_{i+1}$ for all 
  $1\leq i<4n-1$.  Using the reflexivity and transitivity of $\procsim$, we can get a
  (possibly shorter) path alternating between $\procsim$ and $\msgsim$:
  $$
  x=y_{1}\procsim y_{2}\msgsim y_{3}\procsim y_{4} \cdots 
  \msgsim y_{2n'-1}\procsim y_{2n'}=z\,.
  $$
  If $n'>n$ then among $y_{2},y_{4},\ldots,y_{2n+2}$ we find two events $y_{2i},y_{2j}$
  with $i<j$ on the same process.  We have $y_{2i}\sim y_{2j-1}$ and $y_{2i},y_{2j-1}$ on 
  the same process. By $\psi_{2}$ we get $y_{2i}\procsim y_{2j-1}$. Hence,
  we may shorten the path by removing
  $y_{2i},\ldots,y_{2j-1}$.  Repeating this shortening if needed, we end up with a path
  with $n'\leq n$.  In case $n'<n$, since $\procsim$ is reflexive, we may extend the path
  by stuttering the last event $z=y_{2n'}=y_{2n'+1}=\cdots=y_{2n}$. Therefore, $x\sim z$.
  
  We also require that the $Y_m$ variables precisely correspond to messages
  between different $\sim$-classes:
  \[
      \psi'_{2} =
    \paul{\forall x \forall y.\, x\msgrel y \implies 
    \Big( x\not\sim y \Longleftrightarrow
    \bigvee_{m\in\Msg} x\in Y_{m} \wedge y\in Y_{m} \Big) \,.}
  \]
  
  \medskip
  
  Now that the equivalence relation $\sim$ allows us to identify events in the same
  factor, we can state the next constraint: the process order and the message relation
  should induce a partial order on the set of all components.  This is because if
  $\msc\in\msc_1\circ\msc_2\circ\cdots$, process edges or message edges in $\msc$ which
  are not already in one of the components must go from some $\msc_i$ to some
  $\msc_j$ with $i<j$.  We let
  \[
  \bord x y = \neg(x\sim y) \wedge \exists x', y'.\, x \sim x' \land y \sim y' \land
  (x' \msgrel y' \lor x' \mproclt y') \, ,
  \]
  which leads us to the next conjunct in $\psi$: the relation $\bordn$ is acyclic.  Any
  $\bordn$-cycle going through the same process more than twice can be shortened,
  so acyclicity can be expressed as:
  \[
  \psi_{3} = 
  \bigwedge_{1 \le r \le 2n} \lnot \exists x_0, \ldots, x_r.\, x_0 = x_r \land
  \bigwedge_{0 \le i < r} \bord {x_i} {x_{i+1}} \, .
  \]

  Finally, we need to check that every sub-cMSC induced by $\sim$ is indeed in $L$.
  To do so, we define a relativisation $\phir \xi$ of $\xi$ to elements within
  the $\sim$-equivalence class of some event $x$. Message edges with endpoints
  in $Y_m$ should also be ignored, and the label $m$ added.
  This is defined inductively as follows:
  \begin{align*}
    \phir {(\exists y.\, \xi)} & = \exists y. \, y \sim x \land \phir \xi \\
    \phir{m(y)} & = m(y) \lor y \in Y_m \qquad (m \in \Msg)
  \end{align*}
  The other cases are straightforward:
  \begin{align*}
    \phir{(y \le z)} & = y \le z 
    &
    \phir {(y \msgrel z)} & = y \msgrel z
    \\
    \phir{p(y)} & = p(y) \qquad (p \in \Procs)
    &
    \phir{(y\in X)} & = y\in X
    \\
    \phir {(\xi_1 \lor \xi_2)} & = \phir {\xi_1} \lor \phir {\xi_2} 
    &
    \phir {(\lnot \xi)} & = \lnot \phir \xi \,.
  \end{align*}
  The last constraint is that every sub-cMSC induced by $\sim$, with the
  $X_1,\ldots,X_n$-labeling inherited from the full cMSC, satisfies $\phi$ (and therefore
  is in $L$):
  \[
  \psi_{4} = \forall x.\, \phir \phi \, .
  \]
  
  We turn to the proof of correctness, i.e., $L(\Psi)=L^{+}\cup L^{\omega}$.
  
  \paul{First, let $\msc\in\msc_1\circ\msc_2\circ\cdots$ with $\msc_i \in L$ for all 
  $i\geq 1$.
  We show that the cMSC $\msc$ satisfies $\psi$ when choosing an interpretation of the
  second-order variables as explained at the beginning of the proof.
  With the interpretation chosen for $Y_{m}$'s, it is clear that $\psi_{1}$ is satisfied.
  The crucial part is to verify that, with this interpretation, the relation $x\sim y$
  indeed coincide with ``$x,y$ are in the same factor'' of the composition.
  First, since MSCs in $L$ are \emph{connected}, it is easy to see that if $x,y$ are in 
  the same (connected) factor $M_{i}$ then 
  there is a path from $x$ to $y$ of length less than $2n$ using 
  the relation ${\procsim}\cup{\msgsim}$.  Hence $x\sim y$.
  Conversely, assume that $x\procsim y$. 
  From the property of $W$, we deduce that $x,y$ must be in the same factor.
  Now, assume that $x\msgsim y$. 
  From the property of the $Y_{m}$'s, we know that the message edge between $x$ and $y$ 
  is not added by the concatenation. Hence, $x,y$ are in the same factor.
  Therefore, $x\sim y$ implies that $x,y$ are in the same factor.
  From the property of $W$, it is now easy to check that $\psi_{2}$ is satisfied.
  The interpretation chosen for $Y_{m}$'s implies that $\psi'_{2}$ is also satisfied.
  By definition of the concatenation, if $x\msgrel y \lor x\mproclt y$ where $x$ is in 
  $M_{i}$ and $y$ is in $M_{j}$ with $i\neq j$ then $i<j$. We deduce that $\bordn$ is 
  acyclic and $\psi_{3}$ is satisfied.
  Finally, given the interpretation of the $X_{i}$'s and $Y_{m}$'s, the formula 
  $\psi_{4}$ is also satisfied.
  This proves that $L^{+}\cup L^{\omega}\subseteq L(\Psi)$.}
  
  \paul{Conversely, assume that $\msc=(E,\leq,\msgrel,\lambda,\mu)$ satisfies $\Psi$.  
  Consider an interpretation of the set variables $W$, $Y_{m}$'s and $X_{i}$'s such that
  $\psi$ is satisfied.  By $\psi_{2}$, $\sim$ is an equivalence relation on the set of
  events.  Hence, we find sub-cMSCs $\msc_{1},\msc_{2},\ldots$ of $\msc$ induced by the
  equivalence relation $\sim$ and the $Y_{m}$'s.  More precisely, if $e$ is an event in
  $\msc$, then we let $\msc^{e}=([e],\leq,\msgrel,\lambda,\mu')$ be the cMSC with
  $[e]=\{f\mid f\sim e\}$ as set of events, $\leq,\msgrel,\lambda$ restricted to $[e]$, 
  and $f\in\dom(\mu')$ if either $f\in[e]\cap\dom(\mu)$ (in which case $\mu'(f)=\mu(f)$) 
  or if $f\in[e]\cap Y_{m}$ for some $m\in\Msg$ (in which case $\mu'(f)=m$).
  By $\psi_{1}$ these cases are exclusive, hence $\mu'$ is well-defined.
  Notice that by $\psi'_{2}$, if $f\msgrel g$ with $f,g\in[e]$ then
  $f,g\notin\dom(\mu')$.
  By $\psi_{3}$ these sub-cMSCs are
  partially ordered by $\bordn$.  Hence, without loss of generality we may assume that if
  $\bord{\msc_{i}}{\msc_{j}}$ then $i<j$.  By $\psi_{4}$, we have $\msc_{i}\in L$ for
  all $i\geq 1$.  One can then verify
  that $\msc\in\msc_1\circ\msc_2\circ\cdots$.}
  
  Finally, let $\mathsf{finite}=\exists x_{1}\cdots\exists x_{n}\forall y. 
  \bigvee_{i=1}^{n}y\leq x_{i}$ be a sentence defining finite nonempty cMSCs.
  We see that $L^{+}$ and $L^{\omega}$ are respectively defined by 
  $\Psi\wedge\mathsf{finite}$ and $\Psi\wedge\neg\mathsf{finite}$.
  \qed
\end{proof}

\section{Realizability of HMSCs}
\label{sec:realizability}

We now use the closure properties established in the previous section to prove
Theorem~\ref{thm:HMSC-to-CFM}.
This is done, essentially, by applying a standard automata-to-regular-expressions
translation to HMSCs, and observing that the loop-connected assumption
guarantees that iteration is only applied to connected languages.

To that end, we introduce \emph{generalized HMSCs}, which are defined as 
HMSCs except that transitions are labeled by \emph{languages} of finite cMSCs.
More precisely, a generalized HMSC over $\Procs$ is a tuple
$\HMSC = (S, \iota, \Msg, R, F, F_\omega)$ where $S$, $\iota$, $\Msg$, $F$ and
$F_\omega$ are as before, but with a transition function
$R \colon S \times S \to 2^{\fcMSCs}$.

A finite path in $\HMSC$ is a sequence of states
$\rho = s_0 s_1 \ldots s_n \in S^+$ with $n\geq1$.
The set $\mscrun{\rho}$ is defined as the union of all concatenations
$M_1 \circ \cdots \circ M_n$ with $M_i \in R(s_{i-1},s_i)$ for all $1 \le i \le n$.
Infinite paths and associated sets $\mscrun{\rho}$ are defined similarly.
We also write
$\mscrunH \HMSC {s,s'} = \bigcup \{\mscrun{\rho} \mid \rho \textup{ is a finite path of } 
\HMSC \textup{ from } s \textup{ to } s'\}$.
The language $L(\HMSC)$ of a generalized MSC is then defined, as before, as
$L(\HMSC) = \bigcup \{\mscrun{\rho} \mid \rho \textup{ is an accepting path of } \HMSC\}
\cap \MSCs$.
 
A generalized HMSC is called \emph{loop-connected} if for all $s \in S$,
all cMSCs in $\mscrunH \HMSC {s,s}$ are connected.
It is called EMSO-definable if for all states $s$ and~$s'$, there is an EMSO
formula $\phi_{s,s'}$ such that  $R(s,s') = L(\phi_{s,s'})$.

Note that any (loop-connected) HMSC can be seen as a generalized
(loop-connected) HMSC where  for all pairs of states $(s,s')$,
$R(s,s')$ is a finite language.

\smallskip

We can now describe a state elimination procedure for (generalized) HMSCs,
analogous to state elimination for standard automata.
Given a generalized HMSC $\HMSC = (S, \iota, \Msg, R, F, F_\omega)$ and
a state $s \in S \setminus (F \cup F_\omega \cup \{\iota\})$, we define $\HMSC_s = (S \setminus \{s\}, \iota, \Msg, R_s, F, F_\omega)$ as
\[
	R_s(s_1,s_2) = R(s_1,s_2) \cup \big( R(s_1,s) \circ R(s,s_2) \big) \cup 
	\big(R(s_1,s) \circ R(s,s)^+ \circ R(s,s_2)\big) \, .
\]

\begin{lemma}\label{lem:state-elem}
	We have $L(\HMSC) = L(\HMSC_s)$.
	In addition, if $\HMSC$ is loop-connected and EMSO definable,
	then $\HMSC_s$ is also loop-connected and EMSO definable.
\end{lemma}

\begin{proof}
	It is easy to verify that for all states $s_1,s_2 \in S \setminus \{s\}$,
	$\mscrunH \HMSC {s_1,s_2} = \mscrunH {\HMSC_s} {s_1,s_2}$.
	This implies $L(\HMSC) = L(\HMSC_s)$, as well as the fact that
	if $\HMSC$ is loop-connected, then so is $\HMSC_s$.
	
	Together with Lemmas~\ref{lem:emso-union-concatenation} and
	\ref{lem:emso-iteration-connected}, this also implies that if $\HMSC$ is
	both loop-connected and EMSO-definable, then $\HMSC_s$ is as well. \qed
\end{proof}

To prove Theorem~\ref{thm:HMSC-to-CFM}, we only need to show that every
loop-connected HMSC $\HMSC$ can be translated into an EMSO formula
$\phi$ with the same language.
Without loss of generality, we may assume that $\HMSC = (S, \iota, \Msg, R, F, F_\omega)$ is such that $F \cup F_\omega$ is a singleton (as the language of
$\HMSC$ is the finite union of all the languages obtained by keeping only
one accepting state), and that $\iota$ has no incoming transition.

Assume that $F \cup F_\omega = \{s\}$. We successively eliminate all states in
$S \setminus \{\iota,s\}$.
By Lemma~\ref{lem:state-elem}, every intermediate generalized HMSC thus obtained
is loop-connected and EMSO-definable, and such that
$R(s',\iota) = \emptyset$ for all states $s'$, as $\iota$ has no incoming transition in $\HMSC$.
Let $\HMSC' = (\{\iota,s\} , \iota, \Msg, R', F, F_\omega)$
be the HSMC obtained at the end. We then have
\begin{align*}
	L(\HMSC') & = R'(\iota,s) \cup \left(R'(\iota,s) \circ R'(s,s)^+\right) && \text{if } s \in F \\
	L(\HMSC') & = R'(\iota,s) \circ R'(s,s)^\omega && \text{if } s \in F_\omega \, .
\end{align*}
Since we have EMSO formulas for $R'(\iota,s)$ and $R'(s,s)$, 
and $R'(s,s)$ is connected, we can construct an EMSO formula $\phi$ 
such that
\[
	L(\phi) = L(\HMSC') = L(\HMSC) \, .
\]

By Theorem~\ref{thm-EMSO-CFM}, the formula $\phi$ can in turn be translated
into an equivalent CFM.

\section{Satisfiability of HMSCs}
\label{sec:satisfiability}

In this section, we present the proofs related to satisfiability results.  Specifically,
for the satisfiability problem of loop-connected HMSCs, we provide two proofs of
undecidability.
In the first one, we reduce the Post correspondence problem to the satisfiability problem of
a loop-connected \emph{flat} HMSC over three processes. By \emph{flat} we mean that it
contains only self-loops: if there is a path from $s$ to $s'$ and a path from
$s'$ to $s$ then $s=s'$.
In the second proof, we reduce the halting problem of a Turing machine on the empty word 
to the satisfiability problem of a loop-connected HMSC over two processes.
Finally, we also reduce the nonemptiness problem of a two-counter machines to the 
satisfiability problem of HMSCs over four processes and a singleton message set.

\subsection{Undecidability for Loop-Connected Flat HMSCs}\label{sec:sat-flat}

We start with the reduction from the PCP problem. Consider two finite alphabets $A$ and 
$B$ and two morphisms $f,g\colon A^{*}\to B^{*}$. The PCP problem asks whether there is a 
nonempty word $u\in A^{+}$ such that $f(u)=g(u)$. Without loss of generality, we assume 
that $A\cap B=\emptyset$ and $f(A),g(A)\subseteq B^{+}$.

\begin{figure}[tb]
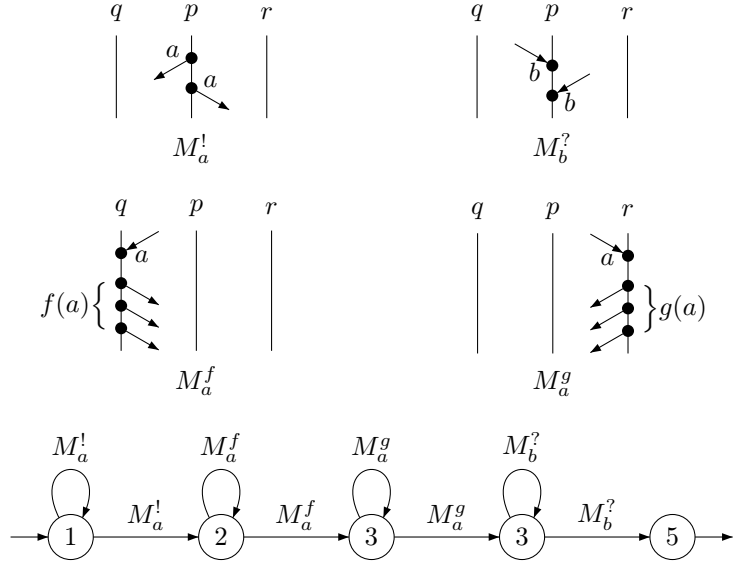

  \centering
  \gusepicture{Ma}\hfil\gusepicture{Mb}
  \\[3ex]
  \gusepicture{Maf}~~~~\hfil~~~~\gusepicture{Mag}
  \\[3ex]
  \gusepicture{flat-HMSC-PCP}
  \caption{Connected and flat HMSC used for the reduction from the PCP problem.}
  \label{fig:flat-hmsc-pcp}
\end{figure}

There are three processes $\Procs=\{p,q,r\}$ and we will only send messages between $p$ 
and $q$ and between $p$ and $r$ (no messages are exchanged between $q$ and $r$).
The set of messages is $\Msg=A\cup B$.
We first define the cMSCs used to label the transitions of the HMSC, see
Figure~\ref{fig:flat-hmsc-pcp}.
For each letter $a\in A$, we define the cMSC $M_{a}^{!}$ which consists only of two 
events $e<f$ on process $p$ which are unmatched sends labeled 
$\lambda(e)=\send{p}{q}$, $\lambda(f)=\send{p}{r}$ and $\mu(e)=\mu(f)=a$.
Assume that $f(a)=b_{1}b_{2}\cdots b_{k}$ with $b_{i}\in B$ for $1\leq i\leq k$.
We also define the cMSC $M_{a}^{f}$ which consists of $k+1$ events $e<f_{1}<\cdots<f_{k}$
on process $q$ which are labeled 
$\lambda(e)=\rec{q}{p}$, $\mu(e)=a$ and 
$\lambda(f_{i})=\send{q}{p}$, $\mu(f_{i})=b_{i}$ for $1\leq i\leq k$.
We define similarly the cMSC $M_{a}^{g}$ which consists of $1+|g(a)|$ events on process 
$r$, first receiving the message $a$ sent by $p$ and then sending to $p$ the sequence 
$g(a)$.
Finally, for each letter $b\in B$ we have a cMSC $M_{b}^{?}$ with two unmatched receives 
on process $p$ with message $b$.
All these cMSMs are \emph{connected} since in each of them only one process has events.

Now, we define a \emph{flat} and connected HMSC $\HMSC$ corresponding to the following
rational expression (see Figure~\ref{fig:flat-hmsc-pcp}).
$$
\Big( \sum_{a\in A} M_{a}^{!} \Big)^{+} \cdot
\Big( \sum_{a\in A} M_{a}^{f} \Big)^{+} \cdot
\Big( \sum_{a\in A} M_{a}^{g} \Big)^{+} \cdot
\Big( \sum_{b\in B} M_{b}^{?} \Big)^{+} \,.
$$
We show that the PCP problem has a solution if and only if the language of $\HMSC$ is 
nonempty.

Let $u=a_{1}a_{2}\cdots a_{n}\in A^{+}$ with $a_{i}\in A$ for $1\leq i\leq n$ be a
solution of the PCP problem, i.e., $f(u)=g(u)=b_{1}b_{2}\cdots b_{m}$ with $b_{i}\in B$
for $1\leq i\leq m$.  Consider the accepting path $\rho$ of $\HMSC$ reading the sequence
$$
M_{a_{1}}^{!}M_{a_{2}}^{!}\cdots M_{a_{n}}^{!}
M_{a_{1}}^{f}M_{a_{2}}^{f}\cdots M_{a_{n}}^{f}
M_{a_{1}}^{g}M_{a_{2}}^{g}\cdots M_{a_{n}}^{g}
M_{b_{1}}^{?}M_{b_{2}}^{?}\cdots M_{b_{m}}^{?} \,.
$$
It is easy to see that $\mscrun{\rho}\cap\MSCs$ is nonempty and contains a unique 
MSC.
Hence the flat and connected HMSC $\HMSC$ has a 
nonemtpy language.

Conversely, assume that $L(\HMSC)$ is nonempty.  Consider an accepting path $\rho$ of
$\HMSC$ such that $\mscrun{\rho}\cap\MSCs\neq\emptyset$.  
Fix some $M\in\mscrun{\rho}\cap\MSCs$.
Let $\sigma=\sigma_{1}\sigma_{2}\sigma_{3}\sigma_{4}$ be the sequence of MSCs read by 
$\rho$ with $\sigma_{1}\in\Big( \sum_{a\in A} M_{a}^{!} \Big)^{+}$,
$\sigma_{2}\in\Big( \sum_{a\in A} M_{a}^{f} \Big)^{+}$,
$\sigma_{3}\in\Big( \sum_{a\in A} M_{a}^{g} \Big)^{+}$, and
$\sigma_{4}\in\Big( \sum_{b\in B} M_{b}^{?} \Big)^{+}$.
Let $u=a_{1}a_{2}\cdots a_{n}\in A^{+}$ be such that
$\sigma_{1}=M_{a_{1}}^{!}M_{a_{2}}^{!}\cdots M_{a_{n}}^{!}$ and let 
$v=b_{1}b_{2}\cdots b_{m}\in B^{+}$ be such that
$\sigma_{4}=M_{b_{1}}^{?}M_{b_{2}}^{?}\cdots M_{b_{m}}^{?}$.
Since all messages sent from $p$ to $q$ must be matched in $M$, we get
$\sigma_{2}=M_{a_{1}}^{f}M_{a_{2}}^{f}\cdots M_{a_{n}}^{f}$, and since all messages sent
back from $q$ to $p$ must be matched in $M$, we get $v=f(u)$.  
Similarly, $\sigma_{3}=M_{a_{1}}^{g}M_{a_{2}}^{g}\cdots M_{a_{n}}^{g}$ and $v=g(u)$.  
Therefore, $u$ is a solution of the PCP problem.

\subsection{Undecidability for Loop-Connected HMSCs over Two Processes}\label{sec:sat-2proc}

Next, we move to the reduction from the halting problem of a Turing machine $\mathcal{M}$
on the empty word to the nonemptiness problem of a loop-connected HMSC $\HMSC$ over
two processes $p$ and $q$.  Intuitively, process $p$ first guesses an initial
configuration $\leftend s_{0} \blank^N$ of $\mathcal{M}$ over the empty word.  Here
$\leftend$ is the left endmarker, $s_{0}$ is the initial state of $\mathcal{M}$ and
$\blank$ is the blank tape symbol.  When gessing this initial configuration, process $p$
uses sufficiently many ($N$) blank symbols for the space needed by the computation until
the halting configuration with state $s_{h}$ is reached.  Process $p$ sends its
configuration $C$ to process $q$, which sends back to $p$ the successor $C'$ of $C$ in the
computation of $\mathcal{M}$.  This is repeated until the halting state is seen by process
$p$, in which case the HMSC $\HMSC$ accepts.  An MSC simulating a run of a Turing machine
is depicted in Figure~\ref{fig:reduction-Turing}.

\begin{figure}[thb]
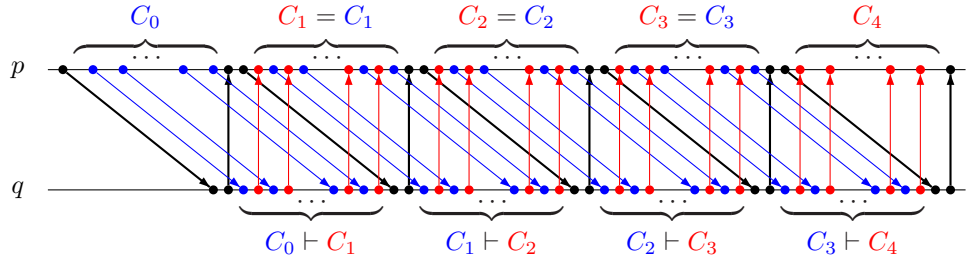

  \centering
  \gusepicture{MrunTM} 
  \caption{Simulation of a run of a Turing machine with an MSC.}
  \label{fig:reduction-Turing}
\end{figure}

Formally, let $S$ be the set of states of the Turing machine $\mathcal{M}$ and $\Gamma$ 
be the set of tape symbols. We have $s_{0},s_{h}\in S$ and $\leftend,\blank\in\Gamma$.
The transition function of $\mathcal{M}$ is given by a subset 
$\Delta\subseteq(\Gamma\times S\times\Gamma)\times(S\cup\Gamma)^{3}$. For instance, 
$\delta=(\alpha s \beta, s' \alpha \gamma)\in\Delta$ means that in state $s$ when reading 
$\beta$ the Turing machine $\mathcal{M}$ goes to state $s'$, replaces $\beta$ with $\gamma$ 
and moves left (a left move is not possible if $\alpha=\leftend$ is the left endmarker).

\begin{figure}[thb]
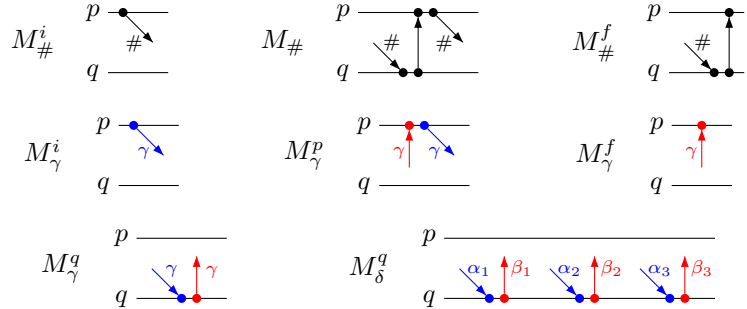

  \centering
  \gusepicture{Msharpi}\hfil\gusepicture{Msharp}\hfil\gusepicture{Msharpf}
  \\[3ex]
  \gusepicture{Mgammai}\hfil\gusepicture{Mgammap}\hfil\gusepicture{Mgammaf}
  \\[3ex]
  \gusepicture{Mgammaq}\hfil\gusepicture{Mdelta}
  \caption{Basic cMSCs for the simulation of a Turing machine.  Here, $\#\notin
  S\cup\Gamma$ is a new (marker) symbol, $\gamma\in S\cup\Gamma$ and
  $\delta=(\alpha_{1}\alpha_{2}\alpha_{3},\beta_{1}\beta_{2}\beta_{3})\in\Delta$.}
  \label{fig:cmsc-reduction-Turing}
\end{figure}

The cMSCs used by the HMSC $\HMSC$ are depicted in 
Figure~\ref{fig:cmsc-reduction-Turing}. Notice that these cMSCs are all \emph{connected}.
We describe the HMSC $\HMSC$ with the following rational expression:
\begin{align*}
  \mathsf{Init} &= M_{\#}^{i} \, M_{\leftend}^{i} \, M_{s_{0}}^{i} \, 
  \big( M_{\blank}^{i} \big)^{+} 
  &
  \mathsf{Succ} &= M_{\#} \cdot \bigg( \sum_{\gamma\in\Gamma} M^{q}_{\gamma} \bigg)^{*} 
  \cdot \bigg( \sum_{\delta\in\Delta} M^{q}_{\delta} \bigg)
  \cdot \bigg( \sum_{\gamma\in\Gamma} M^{q}_{\gamma} \bigg)^{*} 
  \\
  \mathsf{Copy} &= \bigg( \sum_{\gamma\in S\cup\Gamma} M^{p}_{\gamma} \bigg)^{+} 
  &
  \mathsf{Halt} &= \bigg( \sum_{\gamma\in\Gamma} M^{f}_{\gamma} \bigg)^{+} 
  \cdot M_{s_{h}}^{f} 
  \cdot \bigg( \sum_{\gamma\in\Gamma} M^{f}_{\gamma} \bigg)^{*} 
  \cdot M_{\#}^{f} 
  \\
  &&
  \HMSC &= \mathsf{Init} \cdot \big( \mathsf{Succ} \cdot \mathsf{Copy} \big)^{*}
  \cdot \mathsf{Succ} \cdot \mathsf{Halt} 
\end{align*}
Notice that the HMSC $\HMSC$ is loop-connected.  In particular, the loop corresponding to
$\big(\mathsf{Succ}\cdot\mathsf{Copy}\big)^{*}$ is connected since it starts with $M_{\#}$
which connects the two processes.  The MSCs accepted by $\HMSC$ are of the form of
Figure~\ref{fig:reduction-Turing} and correspond to halting computations of the Turing
machine $\mathcal{M}$ on the empty word, i.e., starting from an initial configuration in
$\leftend s_{0} \blank^{+}$.  Therefore, the computation of $\mathcal{M}$ on the empty
word halts if and only if $L(\HMSC)\neq\emptyset$.
This completes the reduction and the undecidability proof for two processes.

\subsection{Undecidability for HMSCs With a Singleton Message Alphabet}\label{sec:sat-sgm}

We reduce nonemptiness (with counter values 0) in two-counter machines to
the corresponding satisfiability problem for HMSCs over $\Procs = \{p_1, p_1', p_2, p_2'\}$
with the singleton message alphabet $\Msg = \{a\}$. As we deal with a singleton set,
we henceforth omit any mention of the message.
Note that our reduction does not allow us to restrict to loop-connected HMSCs.

Just like an HMSC, a two-counter machine $\mathcal{M}$ has a finite state space with
a distinguished initial state and a set of final states (corresponding to $F$
in HMSCs). Moreover, it has two counters, $c_1$ and $c_2$,
whose values range over $\mathbb{N}$. In a transition of $\mathcal{M}$,
any counter can be incremented by one, be decremented by one (provided its
current value is positive), or be tested for zero (the transition can only be taken if the current counter value is 0). The nonemptiness problem asks whether we can reach a final states
when both counters have value 0.

We construct an HMSC $\HMSC$ such that a final state is reachable in
$\mathcal{M}$ iff the language of $\HMSC$ is nonempty.
The idea is that processes $p_1$ and $p_1'$ 
together simulate counter $c_1$, and $p_2$ and $p_2'$ simulate counter $c_2$.
The state-transition structure of $\HMSC$ is exactly the same as that of $\mathcal{M}$.
However, instead of incrementing $c_1$, the HMSC will write
a message into the channel $(p_1, p_1')$, using the cMSC with a
single event labeled $p_1!p_1'$.
Accordingly, decrementing $c_1$ corresponds to appending the cMSC
whose only event is of type $p_1'?p_1$, thus removing a message from the channel.
We proceed analogously for counter $c_2$.
In doing so, the HMSC faithfully simulates the two-counter machine:
at any time in an execution, the number of messages in channel $(p_1, p_1')$ corresponds
to the value of $c_1$ in $\mathcal{M}$, and analogously for channel $(p_2, p_2')$
and counter $c_2$.

It remains to simulate zero tests. A zero test for $c_1$ is simply
replaced, in $\HMSC$, by an MSC containing a single \emph{complete} message from $p_1$ to $p_1'$.
Such a transition contributes to
a run of $\HMSC$ iff, previously, the number of $p_1!p_1'$-events
matches the number of $p_1'?p_1$-events, i.e., iff in the simulated
run of the two-counter machine, the value of $c_1$ is 0. The zero test
for $c_2$ is simulated accordingly.

\section{Conclusion and Future Work}
\label{sec:conclusion}

While we showed that satisfiability problems are undecidable
even for restricted versions of HMSCs,
we solved the synthesis/realizability problem for the class of
loop-connected HMSCs, which allows for unbounded-channel behavior.
This realizability result is orthogonal
to the known result for globally cooperative HMSCs, which define
channel-bounded MSC languages. 
An interesting open problem is whether, over finite MSCs,
loop-connected HMSCs are strictly more expressive than
globally cooperative HMSCs:

\begin{problem}
\label{problem:gcHMSC}
Can every globally cooperative HMSC $\HMSC$ be translated into a loop-connected HMSC $\HMSC'$ 
  such that $L(\HMSC) = L(\HMSC')$?
\end{problem}

Apart from that problem, we leave the following questions for future work:
\begin{itemize}
  \item Does undecidability of satisfiability hold for loop-connected and flat HMSCs with \emph{two} 
  processes?

  \item Is there an interesting subclass of (loop-connected or not) HMSCs beyond
  existentially bounded MSCs with decidable satisfiability/model checking problems?

  \item What are the CFMs that correspond to loop-connected HMSCs?
\end{itemize}


\bibliographystyle{plainurl}
\bibliography{lit}

\end{document}